\documentclass{article}%
\usepackage{authblk}
\usepackage{amsmath}
\usepackage{amsfonts}
\usepackage{amssymb}
\usepackage{graphicx}
\usepackage{geometry}
\usepackage{url}
\usepackage{scalefnt}
\usepackage{cite}
\usepackage[colorlinks=true,citecolor=blue,urlcolor=blue,bookmarks=true,hypertexnames=true]%
{hyperref}
\usepackage{subfigure}%
\providecommand{\U}[1]{\protect\rule{.1in}{.1in}}
\numberwithin{equation}{section}
\newtheorem{theorem}{Theorem}

\newtheorem{proposition}[theorem]{Proposition}

\newenvironment{proof}[1][Proof]{\noindent\textbf{#1.} }{\ \rule{0.5em}{0.5em}}
\begin{document}

\title{Efficient Computation of Periodic Orbits of Forced Rayleigh Equation in the
Framework of Novel Asymptotic Structures}
\author[1] {Aniruddha Palit\thanks{email: mail2apalit@gmail.com}}
\author[2] {Dhurjati Prasad Datta\thanks{Corresponding author; email: dp${_{-}}$datta@yahoo.com}}
\author[3] {Santanu Raut\thanks{email: raut${_{-}}$santanu@yahoo.com}}

\affil[1] {Department of Mathematics, Surya Sen Mahavidyalaya, Siliguri, Pin: 734004, India}
\affil[2] {ICARD, University of North Bengal, Siliguri, Pin: 734013, India}
\affil[3] {Department of Mathematics, Mathabhanga College, Coochbehar, Pin: 736146, India}
\date{}
\maketitle

\begin{abstract}
Higher precision efficient computation of period 1 relaxation oscillations of
strongly nonlinear and singularly perturbed Rayleigh equations with external
periodic forcing is presented. The computations are performed in the context
of conventional renormalization group method (RGM). We demonstrate that
although a slight homotopically modified RGM could generate approximate
periodic orbits that agree qualitatively with the exact orbits, the method,
nevertheless, fails miserably to reduce the large quantitative disagreement
between the theoretically computed results with that of exact numerical
orbits. In the second part of the work we present a novel asymptotic analysis
incorporating SL(2,R) invariant nonlinear deformation of slower time scales,
$t_{n} =\varepsilon^{n}t, \ n\rightarrow\infty, \ \varepsilon<1$, for
asymptotic late time $t$, to a nonlinear time $T_{n}=t_{n}\sigma(t_{n})$,
where the deformation factor $\sigma(t_{n})>0$ respects some well defined
SL(2,R) constraints. Motivations and detailed applications of such nonlinear
asymptotic structures are explained in performing very high accuracy ($>
98\%$) computations of relaxation orbits. Existence of an interesting
condensation and rarefaction phenomenon in connection with dynamically
adjustable scales in the context of a slow-fast dynamical system is explained
and verified numerically.

\end{abstract}

\medskip

\medskip


\section{Introduction}

The aim of the present paper is to formulate an efficient computation scheme
of periodic orbits of a strongly nonlinear oscillator
\cite{nayfey_nonlinear_1995,jordan_nonlinear_1999, Laxman,
bender_advanced_1999}. The importance of high precision computation in applied
mathematics and science need not be overemphasized \cite{Bailey_2012}. The
higher precision quantitatively accurate computation of periodic orbits is
facilitated in the framework of a novel asymptotic analysis \cite{ds15,
palit_comparative_2016, dss18, dsr20 }, so as to allow significant numerical
improvements in the computations of periodic orbits by the conventional
asymptotic techniques such as renormalization group method(RGM)
\cite{Goldenfeld_Lectures_1992}, multiple scale method (MSM)
\cite{nayfey_nonlinear_1995}, homotopy analysis method \cite{liao_beyond_2004,
Cui-NA-2018} etc. As a prototype of strongly nonlinear oscillator, we consider
here singularly perturbed Rayleigh Equation (SRLE)
\cite{nayfey_nonlinear_1995, bender_advanced_1999} equation with an external
periodic excitation%

\begin{equation}
\varepsilon\ddot{x}+\left(  \frac{1}{3}\dot{x}^{3}-\dot{x}\right)
+x=\varepsilon\ F\cos\left(  \Omega\ t\right)  , \label{SFRL Eqn}%
\end{equation}
where dots are used to designate the derivatives with respect to time.
Rayleigh Equation, either regular or singularly perturbed, and a close cousin
of Van der Pol equation \cite{palit_comparative_2016}, is one of the
extensively studied nonlinear oscillatory systems because of its wide
applications in acoustics, physiology and cardiac cycles, solid mechanics,
electronics and nonlinear electrical circuits, musical instruments and many
other different fields \cite{nayfey_nonlinear_1995, jordan_nonlinear_1999,
li_new_2007, stutt_theory_1877}. Singularly perturbed Rayleigh equation is
equivalently related closely to the regularly perturbed Rayleigh equation
(RLE)
\begin{equation}
\ddot{x}+\varepsilon\left(  \frac{1}{3}\dot{x}^{3}-\dot{x}\right)
+x=\varepsilon\ F\cos\left(  \Omega\ t\right)  , \label{RFRL Eqn}%
\end{equation}
for a large nonlinearity parameter $\varepsilon>1$. It is well known that both
the Rayleigh equations $\left(  \text{\ref{SFRL Eqn}}\right)  $ and $\left(
\text{\ref{RFRL Eqn}}\right)  $ with $F=0$ have unique, stable periodic
solutions (orbits), known as the \textit{limit cycle} in the appropriate phase
plane for all $\varepsilon>0$. However, the periodic orbit for singularly
perturbed equation and hence, for larger values of $\varepsilon\gtrapprox1$ in
the regular Rayleigh equation, the periodic cycle is a \emph{relaxation
oscillation}, consisting of slow and fast developing components. Existence of
slow and fast motions \emph{makes traditional analytical techniques
ineffective in an efficient estimation of such relaxation oscillations}
\cite{Xu-Luo-2019, shukla2014new}. Dynamical systems experiencing the
fast-slow motions appear widely in engineering and other applied sciences
\cite{Cui-NA-2018}. The slow-fast periodic motions in such dynamical systems
cannot be easily tackled, because such slow-fast periodic motions need
\emph{many more harmonic terms} to get appropriate approximate solutions.
Usually, the fast movement behaves like an impulsive motion and so grows very
quickly, where as the slow movement is like almost zero velocity movement and
hence the system relaxes very slowly. As will become evident, the present
approach equipped with novel asymptotic quantities, however, would yield such
high precision orbits with much smaller number of harmonic terms only.

To recall, the general study, particularly in the context of high precision
computations, of orbits of the nonlinear ordinary differential equations
(NODE) has always been a challenging task. Exact computation/determination of
analytic solutions of such differential equations is not always possible
because most of them are not generally reducible into exactly integrable form
involving standard functions \cite{Laxman, jordan_nonlinear_1999}. Moreover, a
wide class of NODEs are known to have sensitive dependence on initial
conditions, leading to late time asymptotic unpredictability and chaotic
behaviour. The transition from small time continuity on initial conditions to
the final late time loss of continuity in a nonlinear system generally
proceeds via a universal route, called the period doubling bifurcation route
to chaos. As one or more control (nonlinearity) parameter(s) $\varepsilon$
(say) in the nonlinear system is slowly changed from smaller values to larger
values progressively, the system experiences a sequence of bifurcations in
which a period $n$ orbit is changed suddenly to a period $2n$ orbit at the
bifurcation point $\varepsilon_{n}$ leading finally to chaos \cite{Laxman,
jordan_nonlinear_1999, Perko} . In the absence of exact computability (except
for very special cases, for instance, the Duffing oscillator) of amplitudes,
phases and solutions of isolated periodic orbits (Limit Cycles), along with
precise enumeration of period doubling routes of nonlinear systems,
\emph{higher precision asymptotic} determinations of the periodic $2n$ cycles
are interesting, not only on theoretical ground but also have significant
applications\cite{Cui-NA-2018}.

Over past few decades, various non-perturbative modifications of naive
perturbation method such as Method of Multiple Scales, Method of Boundary
Layer \cite{jordan_nonlinear_1999}, WKB Method \cite{bender_advanced_1999},
Homotopy Averaging Method \cite{cveticanin_periodic_2008}, Homotopy Analysis
Method \cite{liao_beyond_2004}, Variational Iteration Method
\cite{he_asymptotic_2006}, homotopy RG method\cite{liu_renormalization_2017}
etc., have been investigated and advocated widely for faster and efficient
computation of periodic oscillations of strongly nonlinear systems as well as
to singularly perturbed problems, that should yield reasonable fits with
experimental values \emph{for any value} of the control parameter
$\varepsilon>0$. But over time it has become evident that none of the these
asymptotic methods could yield uniformly valid approximate solutions to the
system variables concerned, both for a large control parameter space as well
as for sufficiently large time \cite{shukla2014new}, unless special care and
methods are invented and considered. Recently, Xu and Luo \cite{Xu-Luo-2019}
presented semi-analytic implicit mapping scheme for relaxation oscillations of
forced Van der Pol oscillator. The semi-analytic method is also applied in
various other nonlinear systems such as double pendulum, and many others, see
for instance, {\cite{Luo-Guo-2017}}. In \cite{Luo_book}, Luo presented higher
precision computations of periodic orbits in the context of so called
generalized averaging method. An efficient computation of periodic orbits of
forced Van der Pol-Duffing equation was also presented recently
\cite{Cui-NA-2018} in the context of homotopy analysis method. In all these
cited works \emph{a large number of harmonic terms}, however, are necessary
for the said higher precision computations of periodic orbits. Need for
invoking extensive numerical analysis in attaining highly accurate periodic
orbits, as evidenced in \cite{Luo_book, Xu-Luo-2019, Cui-NA-2018}, based on
different asymptotic modeling and methods, make rooms for further research in
this interesting area of NODE looking \emph{for new and novel theoretical
insights} that would not only \emph{reduce the burden of numerical analysis
and computational time, but might also offer new insights into asymptotic
properties of nonlinear dynamics.}

In the present work, we investigate a modified and improved renormalization
group (RG) method (IRGM) incorporating a novel asymptotic structure, called
$SL(2,%
\mathbb{R}
)$ duality structure, that is introduced and is being investigated previously
by Datta et. al. \cite{ds15, palit_comparative_2016, dss18, dsr20, dr1} in
various nonlinear applications. An application of IRGM and duality structure
was presented in \cite{palit_comparative_2016} in the context of Rayleigh
equation $\left(  \text{\ref{RFRL Eqn}}\right)  $ and Van der Pol equation
(without forcing)
\begin{equation}
\ddot{x}+\varepsilon\ \dot{x}\left(  x^{2}-1\right)  +x=0,\label{VdP}%
\end{equation}
where relative advantage of IRGM over Homotopy analysis method was established
along with high precession analysis of limit cycle orbits (see also
\cite{liu_renormalization_2017}). By efficiency of an asymptotic method we
mean that one could yield $95\%$ or above fits with the experimental values,
only with a lower order approximate computation (1st or 2nd order, in the
perturbative sense), while the other theory would require very high order of
computations to attain that amount of precision (see for instance,
\cite{Luo_book, Cui-NA-2018}). We now apply IRGM equipped with novel
renormalized asymptotic duality structures in obtaining high precision fits
with the theoretical results and that of the experimental results for periodic
orbits of periodically forced Rayleigh oscillators. To emphasize once more,
continual improvements in asymptotic techniques is indeed a very important
research area in applied mathematics, not only for a better theoretical
understanding of periodic or nonperiodic orbits of nonlinear oscillator
problems, but must also have significant engineering and other applications
\cite{Cui-NA-2018, Xu-Luo-2019, Luo-Guo-2017}. A substantial extension of the
$SL(2,%
\mathbb{R}
)$ asymptotic analysis is also formulated recently \cite{dpnw}.

To recall, the improvement of RG Method was made possible in
\cite{palit_comparative_2016} by introduction of $SL(2,%
\mathbb{R}
)$ invariant, renormalized (control) variables, which can be used to ensure an
improved convergence of approximate solutions through continuous deformations
to the exact solution of the nonlinear oscillation problem concerned, with any
desired accuracy level. The introduction of continuous deformation is based on
the novel idea of \textit{nonlinear} time \cite{ds15, dss18, dsr20, dpnw},
incorporating ideas of \emph{asymptotic duality structure} and dynamically
adjustable time scales \cite{dss18,dsr20,dpnw}. It is well known that $\left(
\text{\ref{RFRL Eqn}}\right)  $ has unique limit cycle solution for
$\varepsilon>0$ under the action of near resonant small external periodic
force, i.e. when $F\neq0$ and the frequency $\Omega$ of the external applied
force is close to the frequency of the system or $\Omega\approx1$. In this
article we shall find, for a given $\varepsilon>1$, the values of $F$ and
$\Omega$ for which a significant quantitative difference can be observed
between the exactly computed numerical limit cycle solution and that computed
from RGM. Next we shall improve the initial RG approximation of the exact
solution, using the improved RGM and dynamical time scales along independent
phase space variables, achieving better and faster convergence to the exact
solution. In fact, this article is an extension of the work in
\cite{palit_comparative_2016} for the forced Rayleigh equations with strong
nonlinearity, and hence singularly perturbed, equations $\left(
\text{\ref{SFRL Eqn}}\right)  $ and $\left(  \text{\ref{RFRL Eqn}}\right)  $.
Further, the numerical results presented here demonstrate an interesting
condensation and rarefaction phenomenon in the numerically computed data
points based on the improved RGM calculations and involving the dynamically
adjustable time scales (c.f. Section \ref{Sec Examples}) thus verifying
explicitly the dynamics of slow fast behaviours. Judicious choices of
dynamical time scales would also help reducing the computational complexity in
achieving high precision results.

The paper is organized as follows. In Section \ref{Sec RG}, we review RG
method in the context of singular Rayleigh equation. The failure of usual RG
method is next presented in subsection \ref{Sec RG Failure} when the
singularly perturbed problem is presented as slow-fast system in the phase
plane. In subsection \ref{SRLE}, we demonstrate that the RG method could be
made effective when the singular problem is translated into a regular problem
with a large nonlinearity parameter but, nevertheless, interpreted in the
sense of homotopy deformation \cite{liu_renormalization_2017} involving a
small homotopy parameter $\mu$. The desired result, viz., the limit cycle
solution, is then obtained by continuity by fixing finally $\mu$ to the
limiting value $\mu=1$. In the next two subsections \ref{Sec Alt HRGM} and
\ref{Reg RLE}, we present analogous computations of homotopy aided RG
amplitude and phase equations of regular Rayleigh equation with strong
nonlinearity, and singular Rayleigh equation with various alternate choices of
homotopy deformations respectively. Interestingly, we point out that final
frequency-amplitude equations for all these various homotopy choices are
identical at the order $O(1)$ and hence independent of how the homotopy
perturbation schemes are invoked and implemented. In Section
\ref{Sec Asymptotic Duality}, we present the analysis based on the novel
asymptotic structures and show in Section \ref{Sec Examples}, how the said
structures could successfully yield highly efficient fits with the phase plane
orbits of relaxation oscillations that could have been attained in the
traditional schemes only with very high order of asymptotic computations
involving many more harmonic terms. In the discussion subsection
\ref{Sec Discussion}, we compare our results with those available in current
literature. Finally, we summarize our main conclusions, and remark on future
scope of research in the concluding Section \ref{Sec Conclusion}.

\section{Renormalization Group Method and Singular Rayleigh
Equation\label{Sec RG}}

The Renormalization Group Method (RGM) has a very hallowed history, being
originally formulated in the context of quantum field theory and critical
phenomena \cite{Goldenfeld_Lectures_1992} and later had seen a wide range of
applications in various nonlinear problems, such as solid state physics, fluid
mechanics, cosmology, fractal geometry and many others. The theory of
renormalization group is known to be closely related to the concept of
intermediate asymptotics \cite{barenblatt_scaling_1996}. The applications of
RGM to the NODE was first considered by Chen, Goldenfeld and Oono (CGO) in
\cite{chen_renormalization_1996}. Different authors
\cite{deville_analysis_2008,
palit_comparative_2016,sarkar_renormalization_2011} used the RGM in the
computations of analytic approximations of solutions of various nonlinear
differential systems.

Studies by different authors \cite{chen_renormalization_1996,
sarkar_renormalization_2011} reveal that the RGM has various practical
advantages over other conventional methods such as Boundary Layer theory,
method of Multiple Scales, WKB method etc. In traditional methods, various
\textit{guage functions}, such as fractional power laws or exponential or
logarithmic functions of nonlinearity parameters such as $\varepsilon$ are
usually introduced in an \textit{ad hoc} manner. However, such gauge functions
arise quite \emph{naturally} from the algorithm of RGM. One does not require
to fix such a gauge function from asymptotic matching, or power counting.
However, RGM does have its own limitations. It is shown in
\cite{palit_comparative_2016} that RGM fails to give uniformly valid
approximate solutions for larger values of $\varepsilon\gtrapprox O(1)$. Liu
\cite{liu_renormalization_2017} presented a detailed analysis of failures of
RGM and formulated an improvement of RGM by the so-called Homotopy
Renormalization method (see also \cite{Twin-well-1993}). In Section 2.1, we
show that the perturbative RGM also fails to yield physically relevant results
in the singular Rayleigh system.

The limit cycle solution of the system $\left(  \text{\ref{RFRL Eqn}}\right)
$ with $F=0$ is studied by Palit and Datta in \cite{palit_comparative_2016}
for different values of $\varepsilon>0$. The authors observed that the higher
order computations of this system using RGM fail to improve the classical
results. Further, there does not seem to exist any handle in RGM to ensure
convergence of the approximate solutions to the exact (numerical) solution to
within any specified error bound. Moreover, vary laborious computations needed
generally in the classical RGM for higher order calculations make it quite
improbable for higher order computations. Thus, an improvement of this method
is necessary so that one can ensure convergence of the approximate solutions
to the exact one through a minimal level of computation. With this aim, the
theory of RGM is improved by introduction of some homotopy like deformation
parameters, in association with the framework of an asymptotically nonlinear
time, which we shall discuss in the next section. These control (deformation)
parameters essentially ensured the convergence, using the idea of continuous
deformation of the topological homotopy deformation theory. The introduction
of these control parameters induces the convergence to the approximate
solutions with only a minimal order computation, reducing the computational
complexity of RGM and at the same time improves its efficacy.

The formalism of RGM for the differential equation was introduced by CGO
\cite{chen_renormalization_1996}. This method starts with the naive
perturbative solution of an initial value problem having initial time $t_{0}$.
For periodic oscillations, it involves terms like%
\[
\left(  t-t_{0}\right)  \sin t\text{, }\left(  t-t_{0}\right)  \cos t\text{
etc.}%
\]
which are secular or unbounded as $t$ increases asymptotically. Consequently
the resultant perturbative solution fails to remain periodic and breaks down
for sufficiently large $t$. In order to eliminate this kind of secular terms
CGO introduced an arbitrary time $\tau$ and split the time difference
$t-t_{0}$ as $t-\tau+\tau-t_{0}$ and used RGM to generate a periodic solution
out of this naive perturbative series. One important step, which will be
referred to again later in this section, is the fact that the solution $x$
should not depend on the arbitrary time $\tau$. Therefore, $x$ must be
independent of $\tau$ so that one obtains the renormalization condition%
\begin{equation}
\left.  \frac{dx}{d\tau}\right\vert _{\tau=t}=0\text{.} \label{CGO-RG Eq}%
\end{equation}
DeVille \cite{deville_analysis_2008} introduced an equivalent simplified
version of the RGM which was adopted by Palit and Datta
\cite{palit_comparative_2016}. We shall investigate the solution of the of the
systems $\left(  \text{\ref{SFRL Eqn}}\right)  $ and $\left(
\text{\ref{RFRL Eqn}}\right)  $ using the approach given by DeVille. In this
method the naive perturbative solution of an initial value nonlinear
differential system is derived in terms of complex numbers involving a complex
constant of motion $A$ $($see for detail section 5 in
\cite{deville_analysis_2008}$)$. The initial time $t_{0}$ is taken as
arbitrary with the initial condition $w\left(  t_{0}\right)  $. In the second
step they renormalize the initial condition into $w\left(  t_{0}\right)  $ by
absorbing the time independent and bounded terms in the naive expansion. This
step renormalizes the constant of motion $A$ and generates its counterpart
$\mathcal{A=A}\left(  t_{0}\right)  $ by absorbing the homogeneous parts of
the solution into it. This leaves the solution $x$ having few secular terms
involving%
\[
\left(  t-t_{0}\right)  e^{i\left(  t-t_{0}\right)  },\ \left(  t-t_{0}%
\right)  ^{2}e^{i\left(  t-t_{0}\right)  }\text{ etc.}%
\]
along with their complex conjugates. In the third step the renormalization
condition $\left(  \text{\ref{CGO-RG Eq}}\right)  $ becomes%
\begin{equation}
\left.  \frac{dx}{dt_{0}}\right\vert _{t_{0}=t}=0 \label{RG Eq}%
\end{equation}
which under the transformation%
\begin{equation}
\mathcal{A=}\frac{R}{2}e^{i\theta} \label{Complex to Polar}%
\end{equation}
generates RG flow equations in the amplitude $R$ and the phase $\theta$. To
derive frequency response curve of a forced oscillator, one would generally
require to replace $\theta$ by $\theta(t)=\Omega t+\phi(t)$, that eliminates
the zeroth order constant term, in favour of a term involving the detuning
parameter $\sigma$, in the right hand side of the phase flow equation. Here,
$\Omega=\omega+\sigma\varepsilon$, is the response frequency determined by the
natural frequency $\omega$ and $O(\varepsilon)$ detuning $\sigma$ (c.f.
Section 2.2-4 and Appendix). A brief review of RGM in the context of forced
Rayleigh equation near primary resonance, along with the associated
frequency-amplitude response equation and the associated stability analysis
\cite{nayfey_nonlinear_1995} is given in Appendix.

\subsection{Failure of RGM in Slow-fast system\label{Sec RG Failure}}

To present the case of failure of the classical RGM in the Rayleigh slow-fast
system, we consider, for simplicity, the singular Rayleigh equation $\left(
\text{\ref{SFRL Eqn}}\right)  $ with $F=0$ and $0<\varepsilon<1\ $ which on
differentiating with respect to time has the form
\[
\varepsilon\dddot{x}+\left(  \dot{x}^{2}-1\right)  \ddot{x}+\dot{x}=0.
\]
Assuming%
\[
\dot{x}=y
\]
one gets%
\begin{equation}
\varepsilon\ddot{y}+\left(  y^{2}-1\right)  \dot{y}+y=0,\ 0<\varepsilon<1.
\label{Singular VdP Eq}%
\end{equation}
This is nothing but the singular Van der Pol equation. Since the limit cycle
solution of $\left(  \text{\ref{Singular VdP Eq}}\right)  $ does not depend
upon the initial condition, so without loss of generality, we take the initial
condition as%
\begin{equation}
y\left(  t_{0}\right)  =A_{0}\text{ and }\dot{y}\left(  t_{0}\right)
=B_{0}\text{.} \label{SVdP IC}%
\end{equation}
We write equation (\ref{Singular VdP Eq}) as an autonomous system
\begin{subequations}
\label{SVdP Eq}%
\begin{align}
\dot{y}  &  =z\label{SVdP Eq1}\\
\varepsilon\dot{z}  &  =-\left(  y^{2}-1\right)  z-y \label{SVdP Eq2}%
\end{align}
and the initial condition $\left(  \text{\ref{SVdP IC}}\right)  $ can be
written as%
\end{subequations}
\begin{equation}
y\left(  t_{0}\right)  =A_{0}\text{ and }z\left(  t_{0}\right)  =B_{0}\text{.}
\label{SVdP Auto IC}%
\end{equation}
We introduce the standard boundary layer scaling time $\tau$ using the
transformation
\begin{equation}
t=\varepsilon\ \tau\label{Slow Fast Time}%
\end{equation}
and hence rewrite the autonomous system $\left(  \text{\ref{SVdP Eq}}\right)
$ as
\begin{subequations}
\label{SVdP New}%
\begin{align}
y^{\prime}  &  =\varepsilon z\label{SVdP New Eq1}\\
z^{\prime}  &  =-\left(  y^{2}-1\right)  z-y \label{SVdP New Eq2}%
\end{align}
with the initial conditions $y\left(  \tau_{0}\right)  =A_{0}\text{\ and
\ }z\left(  \tau_{0}\right)  =B_{0}\text{,}$ where $^{\prime}$ represents the
derivative with respect to the variable $\tau$ and $\tau_{0}=\frac{t_{0}%
}{\varepsilon}\text{.}$

Taking
\end{subequations}
\begin{subequations}
\label{SVdP Pert}%
\begin{align}
y\left(  \tau\right)   &  =y_{0}\left(  \tau\right)  +\varepsilon
\ y_{1}\left(  \tau\right)  +\varepsilon^{2}y_{2}\left(  \tau\right)
+\cdots\label{SVdP Pert 1}\\
\text{and }z\left(  \tau\right)   &  =z_{0}\left(  \tau\right)  +\varepsilon
\ z_{1}\left(  \tau\right)  +\varepsilon^{2}z_{2}\left(  \tau\right)
+\cdots\label{SVdP Pert 2}%
\end{align}
from $\left(  \text{\ref{SVdP New}}\right)  $ we get different order relations
as
\end{subequations}
\begin{subequations}
\label{SVdP Pert Eqy}%
\begin{align}
\text{zero-th order}  &  :y_{0}^{\prime}\left(  \tau\right)  =0\qquad
\text{with }y_{0}\left(  \tau_{0}\right)  =A_{0},\label{SVdP Pert Eqy0}\\
\varepsilon\text{ order}  &  :y_{1}^{\prime}\left(  \tau\right)  =z_{0}\left(
\tau\right)  \qquad\text{with }y_{1}\left(  \tau_{0}\right)
=0,\label{SVdP Pert Eqy1}\\
\varepsilon^{2}\text{ order}  &  :y_{2}^{\prime}\left(  \tau\right)
=z_{1}\left(  \tau\right)  \qquad\text{with }y_{2}\left(  \tau_{0}\right)
=0,\label{SVdP Pert Eqy2}\\
&  \cdots\quad\cdots\quad\cdots\quad\cdots\quad\cdots\quad\cdots\nonumber
\end{align}
and
\end{subequations}
\begin{subequations}
\label{SVdP Pert Eqz}%
\begin{align}
\text{zero-th order}  &  :z_{0}^{\prime}+z_{0}\left(  y_{0}^{2}-1\right)
=-y_{0}\qquad\text{with }z_{0}\left(  \tau_{0}\right)  =B_{0}%
,\label{SVdP Pert Eqz0}\\
\varepsilon\text{ order}  &  :z_{1}^{\prime}+z_{1}\left(  y_{0}^{2}-1\right)
=-y_{1}-2y_{0}y_{1}z_{0}\qquad\text{with }z_{1}\left(  \tau_{0}\right)
=0,\label{SVdP Pert Eqz1}\\
\varepsilon^{2}\text{ order}  &  :z_{2}^{\prime}+z_{2}\left(  y_{0}%
^{2}-1\right)  =-y_{2}-z_{0}\left(  y_{1}^{2}+2y_{0}y_{2}\right)  -2y_{0}%
y_{1}z_{1}\qquad\text{with }z_{2}\left(  \tau_{0}\right)
=0.\label{SVdP Pert Eqz2}\\
&  \cdots\quad\cdots\quad\cdots\quad\cdots\quad\cdots\quad\cdots\nonumber
\end{align}
The naive perturbative solutions of $\left(  \text{\ref{SVdP Pert Eqy}%
}\right)  $ and $\left(  \text{\ref{SVdP Pert Eqz}}\right)  $ upto order
$\varepsilon$ are%
\end{subequations}
\begin{align*}
y_{0}\left(  \tau\right)   &  =A_{0},\\
y_{1}\left(  \tau\right)   &  =-\frac{1}{\left(  A_{0}^{2}-1\right)  }\left(
B_{0}+\frac{A_{0}}{A_{0}^{2}-1}\right)  e^{-\left(  A_{0}^{2}-1\right)
\left(  \tau-\tau_{0}\right)  }-\frac{A_{0}}{A_{0}^{2}-1}\left(  \tau-\tau
_{0}\right)  +\frac{1}{\left(  A_{0}^{2}-1\right)  }\left(  B_{0}+\frac{A_{0}%
}{A_{0}^{2}-1}\right)
\end{align*}
and%
\[
z_{0}\left(  \tau\right)  =\left(  B_{0}+\frac{A_{0}}{A_{0}^{2}-1}\right)
e^{-\left(  A_{0}^{2}-1\right)  \left(  \tau-\tau_{0}\right)  }-\frac{A_{0}%
}{A_{0}^{2}-1},
\]%
\begin{align*}
z_{1}\left(  \tau\right)   &  =\dfrac{2A_{0}}{\left(  A_{0}^{2}-1\right)
^{2}}\left(  B_{0}+\dfrac{A_{0}}{A_{0}^{2}-1}\right)  ^{2}e^{-\left(
A_{0}^{2}-1\right)  \left(  \tau-\tau_{0}\right)  }+\frac{A_{0}}{\left(
A_{0}^{2}-1\right)  ^{3}}\left(  1+\frac{2A_{0}^{2}}{\left(  A_{0}%
^{2}-1\right)  }\right)  e^{-\left(  A_{0}^{2}-1\right)  \left(  \tau-\tau
_{0}\right)  }\\
&  +\left(  1+\frac{2A_{0}^{2}}{\left(  A_{0}^{2}-1\right)  }\right)  \frac
{1}{\left(  A_{0}^{2}-1\right)  ^{2}}\left(  B_{0}+\frac{A_{0}}{A_{0}^{2}%
-1}\right)  e^{-\left(  A_{0}^{2}-1\right)  \left(  \tau-\tau_{0}\right)
}-\dfrac{2A_{0}}{\left(  A_{0}^{2}-1\right)  ^{2}}\left(  B_{0}+\dfrac{A_{0}%
}{A_{0}^{2}-1}\right)  ^{2}e^{-2\left(  A_{0}^{2}-1\right)  \left(  \tau
-\tau_{0}\right)  }\\
&  +\left(  1-2A_{0}B_{0}\right)  \frac{1}{\left(  A_{0}^{2}-1\right)
}\left(  B_{0}+\frac{A_{0}}{A_{0}^{2}-1}\right)  e^{-\left(  A_{0}%
^{2}-1\right)  \left(  \tau-\tau_{0}\right)  }\left(  \tau-\tau_{0}\right)
+\frac{A_{0}}{\left(  A_{0}^{2}-1\right)  ^{2}}\left(  1+\frac{2A_{0}^{2}%
}{\left(  A_{0}^{2}-1\right)  }\right)  \left(  \tau-\tau_{0}\right) \\
&  +\dfrac{A_{0}^{2}}{A_{0}^{2}-1}\left(  B_{0}+\dfrac{A_{0}}{A_{0}^{2}%
-1}\right)  e^{-\left(  A_{0}^{2}-1\right)  \left(  \tau-\tau_{0}\right)
}\left(  \tau-\tau_{0}\right)  ^{2}\\
&  -\frac{1}{\left(  A_{0}^{2}-1\right)  ^{2}}\left(  1+\frac{2A_{0}^{2}%
}{\left(  A_{0}^{2}-1\right)  }\right)  \left(  B_{0}+\frac{2A_{0}}{A_{0}%
^{2}-1}\right)  \text{.}%
\end{align*}
We apply classical RGM (See Appendix) by introducing an arbitrary intermediate
time $\lambda$ to split the time interval $\left(  \tau-\tau_{0}\right)  $ as
$\left(  \tau-\lambda\right)  +\left(  \lambda-\tau_{0}\right)  $ and take%
\begin{align*}
A_{0}  &  =Z_{1}\mathcal{A}\left(  \lambda\right) \\
\text{and }B_{0}  &  =Z_{2}\mathcal{B}\left(  \lambda\right)
\end{align*}
where
\begin{align*}
Z_{1}  &  =\sum_{n=0}^{\infty}a_{n}\left(  \tau_{0},\lambda\right)
\varepsilon^{n}=a_{0}+a_{1}\varepsilon+a_{2}\varepsilon^{2}+...\\
\text{and }Z_{2}  &  =\sum_{n=0}^{\infty}b_{n}\left(  \tau_{0},\lambda\right)
\varepsilon^{n}=b_{0}+b_{1}\varepsilon+b_{2}\varepsilon^{2}+...
\end{align*}
with%
\[
a_{0}=1,\ b_{0}=1
\]
so that
\begin{align*}
A_{0}  &  =\mathcal{A}\left(  1+a_{1}\varepsilon+...\right)  ,\\
B_{0}  &  =\mathcal{B}\left(  1+b_{1}\varepsilon+...\right)  .
\end{align*}
The divergent secular terms involving $\left(  \lambda-\tau_{0}\right)  $ are
absorbed in different orders giving the values of $a_{1}$ and $b_{1}$ etc. We
finally obtain the flow equations%
\[
\frac{d\mathcal{A}}{d\tau}=\varepsilon\left[  -\frac{\mathcal{A}}{\left(
\mathcal{A}^{2}-1\right)  }-\frac{d\mathcal{B}}{d\tau}\frac{1}{\left(
\mathcal{A}^{2}-1\right)  }+\frac{2\mathcal{AB}}{\left(  \mathcal{A}%
^{2}-1\right)  ^{2}}\frac{d\mathcal{A}}{d\tau}-\frac{d}{d\tau}\left(
\frac{\mathcal{A}}{\left(  \mathcal{A}^{2}-1\right)  ^{2}}\right)  \right]
+O\left(  \varepsilon^{2}\right)
\]
and%
\[
\frac{d\mathcal{B}}{d\tau}=\frac{d}{d\tau}\left(  \frac{\mathcal{A}}{\left(
\mathcal{A}^{2}-1\right)  }\right)  +\varepsilon\frac{\mathcal{A}}{\left(
\mathcal{A}^{2}-1\right)  ^{2}}\left(  1+2\frac{\mathcal{A}^{2}}%
{\mathcal{A}^{2}-1}\right)  +\varepsilon\frac{d}{d\tau}\left(  \frac
{1}{\left(  \mathcal{A}^{2}-1\right)  ^{2}}\left(  1+2\frac{\mathcal{A}^{2}%
}{\mathcal{A}^{2}-1}\right)  \left(  \mathcal{B}+2\frac{\mathcal{A}%
}{\mathcal{A}^{2}-1}\right)  \right)  +O\left(  \varepsilon^{2}\right)
\]
which in terms of the variable $t$ become%
\[
\frac{d\mathcal{A}}{dt}=-\frac{\mathcal{A}}{\left(  \mathcal{A}^{2}-1\right)
}+O\left(  \varepsilon\right)
\]
and%
\[
\frac{d\mathcal{B}}{dt}=\frac{d}{dt}\left(  \frac{\mathcal{A}}{\left(
\mathcal{A}^{2}-1\right)  }\right)  +\frac{\mathcal{A}}{\left(  \mathcal{A}%
^{2}-1\right)  ^{2}}\left(  1+2\frac{\mathcal{A}^{2}}{\mathcal{A}^{2}%
-1}\right)  +O\left(  \varepsilon\right)
\]
respectively. For steady state (c.f. Appendix) we must have $\frac
{d\mathcal{A}}{dt}=0$ and $\frac{d\mathcal{B}}{dt}=0$ which give
$\mathcal{A}=0$ and $\mathcal{B}=0$. So, the classical RGM fails to give any
limit cycle solution for the singular Van der Pol equation $\left(
\text{\ref{Singular VdP Eq}}\right)  $ \cite{liu_renormalization_2017}. The
failure could also be verified for the forced singular problem with some extra
computations. However, we note, in retrospect, that RGM, extended in the
context of a homotopy deformation, is strong enough to restore the desired
periodic cycles when the singular problem is transformed into a regular
perturbation problem with a large value of the nonlinearity parameter
$\bar{\varepsilon}>1$. We demonstrate this fact in following three subsections
for various choices of homotopy deformations for the original singular problem.

\subsection{Homotopy and RGM: Case 1: SRLE\label{SRLE}}

In Appendix, we detail out briefly the formal RGM computations of amplitude
and phase flow equations for a general periodically forced Rayleigh equation
for any $\varepsilon>0$. Here we consider the singular Rayleigh equation
$\left(  \text{\ref{SFRL Eqn}}\right)  $ with $\varepsilon<1$. It can be
written as%
\begin{equation}
\ddot{x}+\omega^{2}\left(  \frac{1}{3}\dot{x}^{3}-\dot{x}\right)  +\omega
^{2}x=F\cos\left(  \Omega t\right)  \label{SRLE New}%
\end{equation}
where,%
\[
\omega=\sqrt{\frac{1}{\varepsilon}}>1\text{.}%
\]
As noted already, RGM is a powerful asymptotic method that aims to improve
upon the limitations of naive perturbation expansions of periodic responses in
a nonlinear system. It is therefore imperative that the nonlinearity in the
system is weak and the associated non-linearity parameter $\varepsilon$ is
small. In order to study its limit cycle solution we consider homotopically
deformed equation%
\begin{equation}
\ddot{x}+\omega^{2}\mu\left(  \frac{1}{3}\dot{x}^{3}-\dot{x}\right)
+\omega^{2}x=F\mu\cos\left(  \Omega t\right)  \text{ for }\mu\leq1,
\label{HRL Eq}%
\end{equation}
which becomes the Rayleigh equation $\left(  \text{\ref{SRLE New}}\right)  $
for $\mu=1$. Taking%
\begin{equation}
\Omega=\omega+\mu\sigma\label{SRLE Freq}%
\end{equation}
and%
\begin{equation}
x\left(  t\right)  =x_{0}\left(  t\right)  +\mu\ x_{1}\left(  t\right)
+\mu^{2}x_{2}\left(  t\right)  +\cdots\label{SRLE Pert}%
\end{equation}
and following the steps of Appendix, we get the RG flow equations as%
\begin{equation}
\dfrac{dR}{dt}=\mu\left(  \dfrac{R}{2}-\dfrac{R^{3}}{8}\omega^{2}\right)
\omega^{2}-\dfrac{F}{2\omega}\ \mu\sin\theta+O\left(  \mu^{2}\right)
\label{SRLE RG Eq1}%
\end{equation}
and
\begin{equation}
\dfrac{d\theta}{dt}=\omega-\dfrac{F}{2R\omega}\mu\cos\theta+O\left(  \mu
^{2}\right)  \text{.} \label{SRLE RG Eq2}%
\end{equation}
The renormalized solution after taking the limit $t\rightarrow\infty$ and
eliminating the secular terms is%
\begin{equation}
x\left(  t\right)  =R\left(  t\right)  \cos\left(  \theta\left(  t\right)
\right)  +\frac{1}{96}\mu\omega^{3}R^{3}\left(  t\right)  \sin\left(
3\theta\left(  t\right)  \right)  \label{SRLE RG Sol}%
\end{equation}
where, $R\left(  t\right)  $ and $\theta\left(  t\right)  $ are obtained from
the solution of the RG flow equations. Invoking continuity of deformed system
and letting $\mu\rightarrow1^{-}$ in $\left(  \text{\ref{SRLE RG Eq1}}\right)
$, $\left(  \text{\ref{SRLE RG Eq2}}\right)  $ and $\left(
\text{\ref{SRLE RG Sol}}\right)  $ we get,%
\begin{align}
\dfrac{dR}{dt}  &  =\left(  \dfrac{R}{2}-\dfrac{R^{3}}{8}\omega^{2}\right)
\omega^{2}-\dfrac{F}{2\omega}\sin\theta\label{SRLE RG New Eq1}\\
\dfrac{d\theta}{dt}  &  =\omega-\dfrac{F}{2R\omega}\cos\theta
\label{SRLE RG New Eq2}%
\end{align}
and
\begin{equation}
x\left(  t\right)  =R\left(  t\right)  \cos\left(  \theta\left(  t\right)
\right)  +\frac{1}{96}\omega^{3}R^{3}\left(  t\right)  \sin\left(
3\theta\left(  t\right)  \right)  . \label{SRLE RG New Sol}%
\end{equation}

As, by continuity, $\mu$ attains the limiting value $\mu=1$, one gets by
writing%
\begin{equation}
\theta=\phi+\Omega\ t=\phi+\left(  \omega+\sigma\right)  t \label{SRLE Phase}%
\end{equation}
the flow equations as,%
\begin{equation}
\dfrac{dR}{dt}=\left(  \dfrac{R}{2}-\dfrac{R^{3}}{8}\omega^{2}\right)
\omega^{2}-\dfrac{F}{2\omega}\sin\left(  \phi+\Omega\ t\right)
\label{SRLE RG Revised Eq1}%
\end{equation}
and
\begin{equation}
\dfrac{d\phi}{dt}+\sigma=-\dfrac{F}{2R\omega}\cos\left(  \phi+\Omega
\ t\right)  \text{.} \label{SRLE RG Revised Eq2}%
\end{equation}
For steady state motion $\frac{dR}{dt}=0$ and $\frac{d\phi}{dt}=0$ so that
$\left(  \text{\ref{SRLE RG Revised Eq1}}\right)  $ and $\left(
\text{\ref{SRLE RG Revised Eq2}}\right)  $ give%
\[
\left(  \dfrac{R}{2}-\dfrac{R^{3}}{8}\omega^{2}\right)  \omega=\dfrac
{F}{2\omega^{2}}\sin\left(  \phi+\Omega\ t\right)
\]
and%
\[
\frac{R\sigma}{\omega}=-\dfrac{F}{2\omega^{2}}\cos\left(  \phi+\Omega
\ t\right)  .
\]
Squaring and then adding we get,%
\[
\dfrac{R^{2}\omega^{2}}{4}\left(  1-\dfrac{R^{2}}{4}\omega^{2}\right)
^{2}+\frac{R^{2}\sigma^{2}}{\omega^{2}}=\dfrac{F^{2}}{4\omega^{4}}.
\]
Taking
\[
\rho=\frac{R^{2}\omega^{2}}{4},\ \sigma_{1}=\frac{\sigma}{\omega^{2}%
}=\varepsilon\sigma\text{ and }k=\frac{F}{\omega^{2}}%
\]
we get the frequency-amplitude response equation%
\begin{equation}
\rho\left(  1-\rho\right)  ^{2}+4\rho\sigma_{1}^{2}=\dfrac{k^{2}}{4},
\label{SRLE Freq-response Eq}%
\end{equation}
which is identical to the equation $\left(  \text{\ref{Freq-Response Eq}%
}\right)  $ as described in the Appendix and so the region giving stable limit
cycle must satisfy both the inequalities%
\[
\rho>\frac{1}{2}\text{ and }\Delta>0
\]
in $\rho\sigma_{1}$ plane, where%
\[
\Delta=\frac{1}{4}\left(  1-4\rho+3\rho^{2}\right)  +\sigma_{1}^{2}\text{.}%
\]
Taking%
\begin{equation}
k=0.5,\ \omega=1.15,\ \sigma_{1}=0.07 \label{SRLE Parameter Val}%
\end{equation}
so that%
\[
F=0.5\omega^{2}%
\]
$\left(  \text{\ref{SRLE Freq-response Eq}}\right)  $ gives three values of
$\rho$ as%
\[
\rho=7.0777\times10^{-2},\ 0.74685\text{ and }1.1824\text{.}%
\]
We take the largest value $\rho=1.1824$ so that%
\[
\Delta=0.12105>0
\]
i.e., our choice of parameters in $\left(  \text{\ref{SRLE Parameter Val}%
}\right)  $ satisfy both the convergence criteria. The limit cycle from RG
solution $\left(  \text{\ref{SRLE RG New Sol}}\right)  $ is represented by
dotted line along with the numerically computed (exact) limit cycle
represented by solid line in the Figure \ref{f:1a} with the values of the
parameters given by $\left(  \text{\ref{SRLE Parameter Val}}\right)  $.

\begin{figure}[h]
\begin{center}
\centering\subfigure[]{\includegraphics[height=0.3\linewidth]{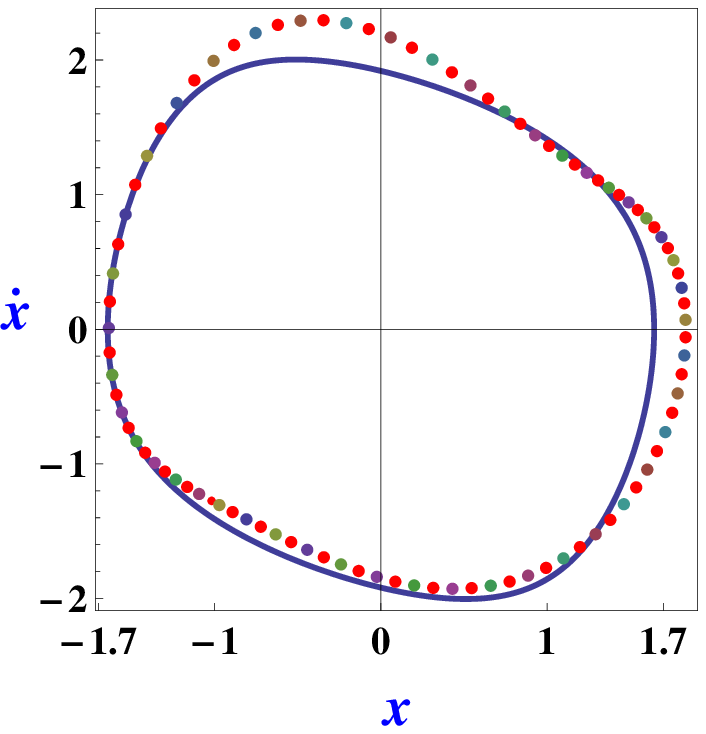}\label{f:1a}}\qquad
\qquad
\subfigure[]{\includegraphics[height=0.3\linewidth]{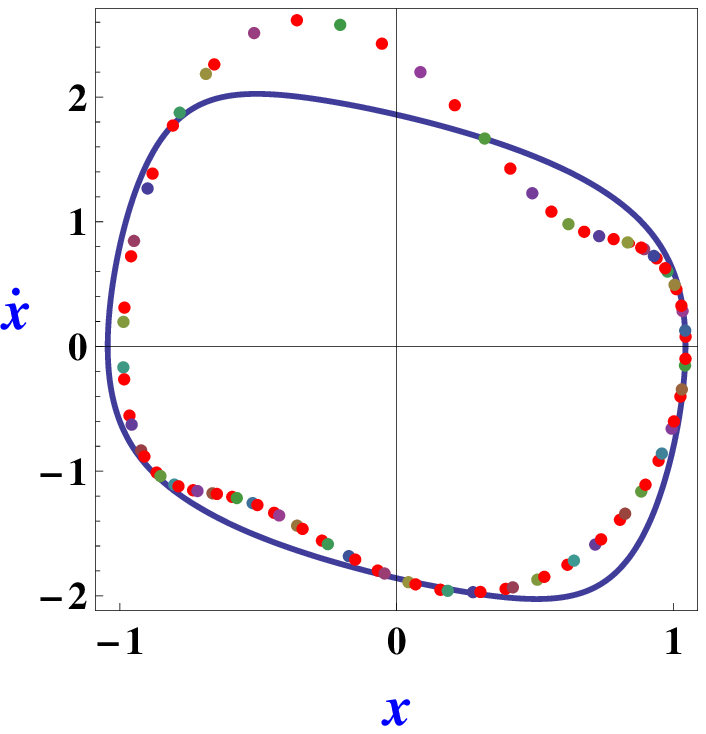}\label{f:1b}}
\end{center}
\caption{{Comparison of limit cycles given by the approximate solution
}$\left(  \text{\ref{SRLE RG New Sol}}\right)  $ using RGM (in dotted line){
with the exact numerical limit cycle (in solid line) for }$\left(  a\right)  $
$k=0.5\ \left(  \text{or }F=0.5\omega^{2}\right)  ,\ \omega=1.15\ \left(
\text{or }\varepsilon=0.75614\right)  ,\ \sigma_{1}=0.07$ and $\left(
b\right)  $ $k=0.4\ \left(  \text{or }F=0.4\omega^{2}\right)  ,\ \omega
=2\ \left(  \text{or }\varepsilon=0.25\right)  ,\ \sigma_{1}=0.01.$}%
\label{f:1}%
\end{figure}

Similarly, for%
\begin{equation}
k=0.4,\ \omega=2,\ \sigma_{1}=0.01 \label{SRLE Parameter Val2}%
\end{equation}
so that%
\[
F=0.4\omega^{2}%
\]
we have $\varepsilon=0.25$ and $\left(  \text{\ref{SRLE Freq-response Eq}%
}\right)  $ gives three values of $\rho$ as%
\[
\rho=4.3722\times10^{-2},\ 0.77347\text{ and }1.1828
\]
We take the largest value $\rho=1.1824$ so that%
\[
\Delta=0.12105>0
\]
i.e., our choice of parameters in $\left(  \text{\ref{SRLE Parameter Val2}%
}\right)  $ satisfy both the convergence criteria. The limit cycle from RG
solution $\left(  \text{\ref{SRLE RG New Sol}}\right)  $ is represented by
dotted line along with the numerically computed (exact) limit cycle
represented by solid line in the Figure \ref{f:1b} with the values of the
parameters given by $\left(  \text{\ref{SRLE Parameter Val2}}\right)  $.

To summarize, homotopically extended RGM is successful in reproducing
qualitatively reasonable relaxation oscillation solutions when the singular
problem is transformed into a regular problem but with a large value of
nonlinearity parameter. In the next subsection, we show that almost analogous
results are also recovered, with only minor differences in the explicit forms
of the approximate solutions, by implementing an alternative homotopy
variation. In subsection \ref{Reg RLE}, we consider the forced Rayleigh
problem with strong nonlinearity in the context of homotopy RGM. It is
heartening to note that the limiting homotopy RGM solutions matches exactly
with the formal RGM solution of Appendix, establishing, in retrospect, the
strength of classical RGM even yielding qualitatively reasonable periodic
solutions in a strongly nonlinear oscillation.

\subsection{Homotopy and RGM: Case 2: SRLE, Alternate homotopy
deformation\label{Sec Alt HRGM}}

Let us consider the singular Rayleigh equation $\left(  \text{\ref{SFRL Eqn}%
}\right)  $ with $\varepsilon<1$, but present in the alternative form
\begin{equation}
\varepsilon\ddot{x}(t)+\varepsilon x\left(  t\right)  +\left\{  \left(
\frac{1}{3}\dot{x}(t)^{3}-\dot{x}(t)\right)  +\left(  1-\varepsilon\right)
x(t)\right\}  =F\varepsilon\cos(\Omega t),\ \varepsilon<1\text{.}
\label{SRLE Alt Eq}%
\end{equation}
Motivation of this choice comes from the salient fact that the success of RGM
in recovering the qualitative features of periodic oscillations basically
rests on the availability of a zeroth order harmonic term in the perturbative
expansion of the desired solution. Since the non-linearity term in the above
equation is still large, we consider, instead the following homotopically
extended system to study the limit cycle solution
\begin{equation}
\varepsilon\ddot{x}(t)+\varepsilon x\left(  t\right)  +\mu\left\{  \left(
\frac{1}{3}\dot{x}(t)^{3}-\dot{x}(t)\right)  +\left(  1-\varepsilon\right)
x(t)\right\}  =F\varepsilon\mu\cos(\Omega t),\ \varepsilon<1,
\label{SRLE Alt Hom Eq}%
\end{equation}
which becomes the forced singular problem $\left(  \text{\ref{SRLE Alt Eq}%
}\right)  $ for $\mu=1$. Next, choosing, as above,
\begin{equation}
\Omega=1+\mu\sigma\text{.} \label{SRLE Alt Freq}%
\end{equation}
one expands limit cycle solution as%
\begin{equation}
x\left(  t\right)  =x_{0}\left(  t\right)  +\mu\ x_{1}\left(  t\right)
+\mu^{2}x_{2}\left(  t\right)  +... \label{SRLE Alt Pert}%
\end{equation}
so that the RG flow equations are (Appendix)%
\begin{align}
\frac{dR}{dt}  &  =\frac{1}{\varepsilon}\mu\left(  \frac{R}{2}-\frac{R^{3}}%
{8}\right)  -\frac{F}{2}\mu\sin\theta+O\left(  \mu^{2}\right)
\label{SRLE Alt RG Eq1}\\
\frac{d\theta}{dt}  &  =1-\frac{\mu}{2R\varepsilon}\left(  F\varepsilon
\cos\theta+R\left(  \varepsilon-1\right)  \right)  +O\left(  \mu^{2}\right)
\text{.} \label{SRLE Alt RG Eq2}%
\end{align}
The renormalized solution after taking the limit $t\rightarrow\infty$ and
eliminating the secular terms is%
\begin{equation}
x\left(  t\right)  =R\left(  t\right)  \cos\left(  \theta\left(  t\right)
\right)  +\frac{1}{96\varepsilon}\mu R^{3}\left(  t\right)  \sin\left(
3\theta\left(  t\right)  \right)  \text{.} \label{SRLE Alt RG Sol}%
\end{equation}
Finally, putting $\mu=1$, by continuity, in $\left(
\text{\ref{SRLE Alt RG Eq1}}\right)  $, $\left(  \text{\ref{SRLE Alt RG Eq2}%
}\right)  $ and $\left(  \text{\ref{SRLE Alt RG Sol}}\right)  $ we get,%
\begin{align}
\frac{dR}{dt}  &  =\frac{1}{\varepsilon}\left(  \frac{R}{2}-\frac{R^{3}}%
{8}\right)  -\frac{F}{2}\sin\theta\label{SRLE Alt RG New Eq1}\\
\frac{d\theta}{dt}  &  =\frac{1}{2}+\frac{1}{2\varepsilon}-\frac{1}{2R}%
F\cos\theta\label{SRLE Alt RG New Eq2}%
\end{align}
and%
\begin{equation}
x\left(  t\right)  =R\left(  t\right)  \cos\left(  \theta\left(  t\right)
\right)  +\frac{1}{96\varepsilon}R^{3}\left(  t\right)  \sin\left(
3\theta\left(  t\right)  \right)  \text{.} \label{SRLE Alt RG New Sol}%
\end{equation}
Taking%
\[
\theta=\phi+\left(  \frac{1}{2}+\frac{1}{2\varepsilon}+\sigma\right)  \ t
\]
from $\left(  \text{\ref{SRLE Alt RG New Eq1}}\right)  $ and $\left(
\text{\ref{SRLE Alt RG New Eq2}}\right)  $ we get,%
\begin{align}
\frac{dR}{dt}  &  =\frac{1}{\varepsilon}\left(  \frac{R}{2}-\frac{R^{3}}%
{8}\right)  -\frac{F}{2}\sin\left(  \phi+\left(  \frac{1}{2}+\frac
{1}{2\varepsilon}+\sigma\right)  \ t\right) \label{SRLE Alt RG Revised Eq1}\\
\frac{d\phi}{dt}+\sigma &  =-\frac{1}{2R}F\cos\left(  \phi+\left(  \frac{1}%
{2}+\frac{1}{2\varepsilon}+\sigma\right)  \ t\right)  \text{.}
\label{SRLE Alt RG Revised Eq2}%
\end{align}
For steady state motion $\frac{dR}{dt}=0$ and $\frac{d\phi}{dt}=0$ so that the
above equations give%
\begin{align*}
\left(  \frac{R}{2}-\frac{R^{3}}{8}\right)   &  =\frac{F\varepsilon}{2}%
\sin\left(  \phi+\left(  \frac{1}{2}+\frac{1}{2\varepsilon}+\sigma\right)
\ t\right) \\
R\varepsilon\sigma &  =-\frac{1}{2}F\varepsilon\cos\left(  \phi+\left(
\frac{1}{2}+\frac{1}{2\varepsilon}+\sigma\right)  \ t\right)  \text{.}%
\end{align*}
Squaring and then adding we get,%
\[
\frac{R^{2}}{4}\left(  1-\frac{R^{2}}{4}\right)  ^{2}+R^{2}\varepsilon
^{2}\sigma^{2}=\frac{F^{2}\varepsilon^{2}}{4}.
\]
Taking%
\[
\rho=\frac{R^{2}}{4},\ \sigma_{1}=\varepsilon\sigma\text{ and }k=F\varepsilon
\]
we get the frequency-amplitude response equation%
\begin{equation}
\rho\left(  1-\rho\right)  ^{2}+4\rho\sigma_{1}^{2}=\frac{k^{2}}{4},
\label{SRLE Alt Freq-response Eq}%
\end{equation}
which is identical to the equation $\left(  \text{\ref{Freq-Response Eq}%
}\right)  $ as described in the Appendix and so the region giving stable limit
cycle must satisfy both the inequalities%
\[
\rho>\frac{1}{2}\text{ and }\Delta>0
\]
in $\rho\sigma_{1}$ plane, where%
\[
\Delta=\frac{1}{4}\left(  1-4\rho+3\rho^{2}\right)  +\sigma_{1}^{2}\text{.}%
\]
Taking%
\begin{equation}
k=0.5,\ \varepsilon=0.75,\ \sigma_{1}=0.07 \label{SRLE Alt Parameter Val}%
\end{equation}
so that%
\[
F=\frac{0.5}{\varepsilon}%
\]
$\left(  \text{\ref{SRLE Alt Freq-response Eq}}\right)  $ gives three values
of $\rho$ as%
\[
\rho=7.0777\times10^{-2},\ 0.74685\text{ and }1.1824\text{.}%
\]
If we take the largest value $\rho=1.1824$ then%
\[
\Delta=0.12105>0\text{.}%
\]
Therefore, this choice of the parameter values given by $\left(
\text{\ref{SRLE Alt Parameter Val}}\right)  $ satisfy the convergence
criteria. The limit cycle from RG solution $\left(
\text{\ref{SRLE Alt RG New Sol}}\right)  $ is represented by dotted line along
with the numerically computed (exact) limit cycle represented by solid line in
the Figure \ref{f:2a} with the values of the parameters given by $\left(
\text{\ref{SRLE Alt Parameter Val}}\right)  $. \begin{figure}[h]
\begin{center}
\centering\subfigure[]{\includegraphics[height=0.3\linewidth]{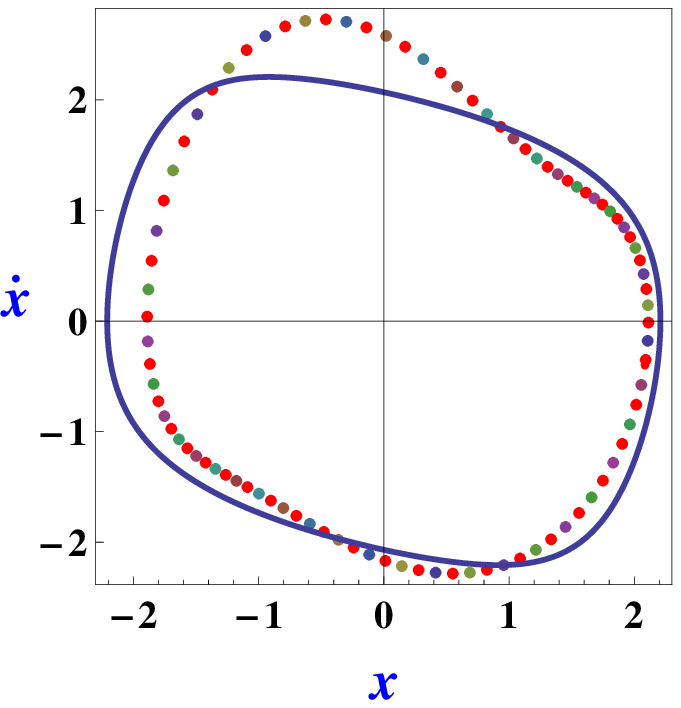}\label{f:2a}}\qquad
\qquad
\subfigure[]{\includegraphics[height=0.3\linewidth]{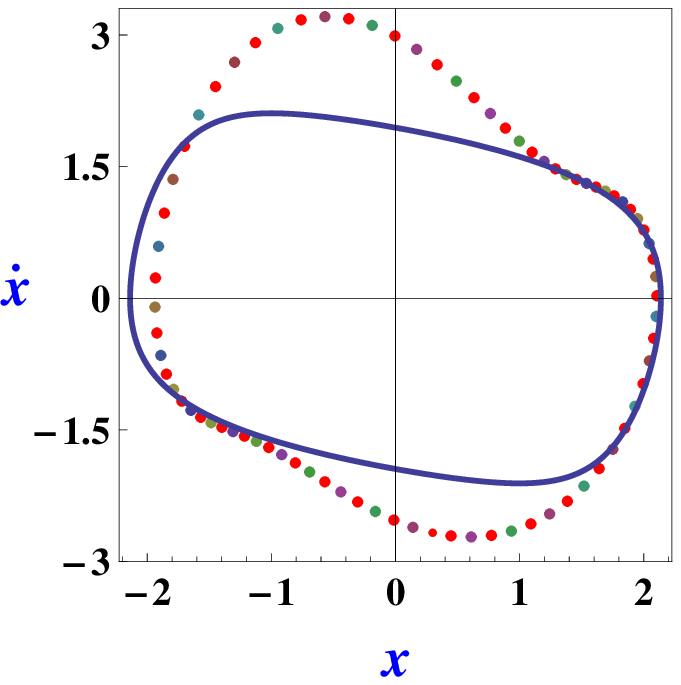}\label{f:2b}}
\end{center}
\caption{{Comparison of limit cycles given by the approximate solution
}$\left(  \text{\ref{SRLE Alt RG New Sol}}\right)  ${ using RGM (in dotted
line) with the exact numerical limit cycle (in solid line) for }$\left(
a\right)  ${ }$k=0.5,\ \varepsilon=0.75,\ \sigma_{1}=0.07$ and $\left(
b\right)  $ $k=0.4,\ \varepsilon=0.6,\ \sigma_{1}=0.07.$}%
\label{f:2}%
\end{figure}Similarly, taking%
\begin{equation}
k=0.4,\ \varepsilon=0.6,\ \sigma_{1}=0.07 \label{SRLE Alt Parameter Val2}%
\end{equation}
so that%
\[
F=\frac{0.4}{\varepsilon}%
\]
$\left(  \text{\ref{SRLE Alt Freq-response Eq}}\right)  $ gives three values
of $\rho$ as%
\[
4.2737\times10^{-2},\ 0.83109\text{ and }1.1262.
\]
Taking the largest value $\rho=1.1262$ we get%
\[
\Delta=7.9945\times10^{-2}>0
\]
i.e., this choice of the parameter values given by $\left(
\text{\ref{SRLE Alt Parameter Val2}}\right)  $ satisfy the convergence
criteria. The limit cycle from RG solution $\left(
\text{\ref{SRLE Alt RG New Sol}}\right)  $ is represented by dotted line along
with the numerically computed (exact) limit cycle represented by solid line in
the Figure \ref{f:2b} with the values of the parameters given by $\left(
\text{\ref{SRLE Alt Parameter Val2}}\right)  $.

To conclude this section, we note that two alternative choices of homotopy
deformations qualitatively almost identical periodic solutions c.f
(\ref{SRLE RG New Sol}) and (\ref{SRLE Alt RG New Sol}). Two solutions have
almost identical forms except for minor differences in coefficients of the
second terms.

\subsection{Homotopy and RGM: Case 3: Regular RLE for all $\varepsilon
$\label{Reg RLE}}

In this section we consider $\left(  \text{\ref{RFRL Eqn}}\right)  $ for all
$\varepsilon>0$. In order to study its limit cycle solution we consider the
associated homotopy equation%
\begin{equation}
\ddot{x}(t)+\varepsilon\mu\left(  \frac{1}{3}\dot{x}(t)^{3}-\dot{x}(t)\right)
+x(t)=F\epsilon\mu\cos(\Omega t), \label{RFRL New Eq}%
\end{equation}
which becomes the forced Rayleigh equation $\left(  \text{\ref{RFRL Eqn}%
}\right)  $ for $\mu=1$. Taking%
\begin{equation}
\Omega=1+\mu\varepsilon\sigma\label{RFRL Freq}%
\end{equation}
and%
\begin{equation}
x\left(  t\right)  =x_{0}\left(  t\right)  +\mu\ x_{1}\left(  t\right)
+\mu^{2}x_{2}\left(  t\right)  +... \label{RFRL Pert}%
\end{equation}
we get the RG flow equations%
\begin{align}
\frac{dR}{dt}  &  =\mu\varepsilon\left(  \frac{R}{2}-\frac{R^{3}}{8}\right)
-\frac{F}{2}\mu\varepsilon\sin\theta+O\left(  \mu^{2}\right)
\label{RFRL RG Eq1}\\
\frac{d\theta}{dt}  &  =1-\mu\frac{F}{2R}\varepsilon\cos\theta+O\left(
\mu^{2}\right)  \label{RFRL RG Eq2}%
\end{align}
and the renormalized solution after taking the limit $t\rightarrow\infty$ and
eliminating the secular terms is%
\begin{equation}
x\left(  t\right)  =R\left(  t\right)  \cos\left(  \theta\left(  t\right)
\right)  +\frac{1}{96}\varepsilon\mu R^{3}\left(  t\right)  \sin\left(
3\theta\left(  t\right)  \right)  \label{RFRL RG Sol}%
\end{equation}
where, $R\left(  t\right)  $ and $\theta\left(  t\right)  $ are the solutions
of $\left(  \text{\ref{RFRL RG Eq1}}\right)  $ and $\left(
\text{\ref{RFRL RG Eq2}}\right)  $. Putting $\mu=1$, by continuity, in
$\left(  \text{\ref{RFRL RG Eq1}}\right)  $, $\left(  \text{\ref{RFRL RG Eq2}%
}\right)  $ and $\left(  \text{\ref{RFRL RG Sol}}\right)  $ we get,%
\begin{align}
\frac{dR}{dt}  &  =\varepsilon\left(  \frac{R}{2}-\frac{R^{3}}{8}\right)
-\frac{F}{2}\varepsilon\sin\theta\label{RFRL RG New Eq1}\\
\frac{d\theta}{dt}  &  =1-\frac{F}{2R}\varepsilon\cos\theta
\label{RFRL RG New Eq2}%
\end{align}
and%
\begin{equation}
x\left(  t\right)  =R\left(  t\right)  \cos\left(  \theta\left(  t\right)
\right)  +\frac{1}{96}\varepsilon R^{3}\left(  t\right)  \sin\left(
3\theta\left(  t\right)  \right)  \text{.} \label{RFRL RG New Sol}%
\end{equation}
As $\mu=1$ taking%
\[
\theta=\phi+\Omega\ t=\phi+\left(  1+\varepsilon\sigma\right)  t
\]
in $\left(  \text{\ref{RFRL RG New Eq1}}\right)  $ and $\left(
\text{\ref{RFRL RG New Eq2}}\right)  $ we get,%
\begin{align}
\frac{dR}{dt}  &  =\varepsilon\left(  \frac{R}{2}-\frac{R^{3}}{8}\right)
-\frac{F}{2}\varepsilon\sin\left(  \phi+\Omega\ t\right)
\label{RFRL RG Revised Eq1}\\
\frac{d\phi}{dt}+\varepsilon\sigma &  =-\frac{F}{2R}\varepsilon\cos\left(
\phi+\Omega\ t\right)  \text{.} \label{RFRL RG Revised Eq2}%
\end{align}
For steady state motion $\frac{dR}{dt}=0$ and $\frac{d\phi}{dt}=0$, so that
the above equations give%
\begin{align*}
\left(  \frac{R}{2}-\frac{R^{3}}{8}\right)   &  =\frac{F}{2}\sin\theta\\
R\sigma &  =-\frac{F}{2}\cos\theta\text{.}%
\end{align*}
Squaring and then adding we get,%
\[
\frac{R^{2}}{4}\left(  1-\frac{R^{2}}{4}\right)  ^{2}+R^{2}\sigma^{2}%
=\frac{F^{2}}{4}\text{.}%
\]
Taking%
\[
\rho=\frac{R^{2}}{4}\text{ and }k=F
\]
we get the frequency-amplitude response equation as,%
\begin{equation}
\rho\left(  1-\rho\right)  ^{2}+4\rho\sigma^{2}=\frac{k^{2}}{4},
\label{RFRL Freq-response Eq}%
\end{equation}
which is again identical to the equation $\left(  \text{\ref{Freq-Response Eq}%
}\right)  $ as described in the Appendix and so the region giving stable limit
cycle must satisfy both the inequalities%
\[
\rho>\frac{1}{2}\text{ and }\Delta>0
\]
in $\rho\sigma$ plane, where%
\[
\Delta=\frac{1}{4}\left(  1-4\rho+3\rho^{2}\right)  +\sigma^{2}\text{.}%
\]
Taking%
\begin{equation}
k=F=0.5,\ \varepsilon=1.3,\ \sigma=0.07 \label{RFRL2 Parameter Val}%
\end{equation}
$\left(  \text{\ref{RFRL Freq-response Eq}}\right)  $ gives three values of
$\rho$ as%
\[
\rho=7.0777\times10^{-2},\ 0.74685,\ 1.1824\text{.}%
\]
If we take the largest value $\rho=1.1824$ then%
\[
\Delta=0.12105>0\text{.}%
\]
Therefore, this choice of the values of the parameters given by $\left(
\text{\ref{RFRL2 Parameter Val}}\right)  $ satisfy both the convergence
criteria. The limit cycle from RG solution $\left(
\text{\ref{RFRL RG New Sol}}\right)  $ is represented by dotted line along
with the numerically computed (exact) limit cycle represented by solid line in
the Figure \ref{f:3a} with the values of the parameters given by $\left(
\text{\ref{RFRL2 Parameter Val}}\right)  $.

\begin{figure}[h]
\begin{center}
\centering\subfigure[]{\includegraphics[height=0.3\linewidth]{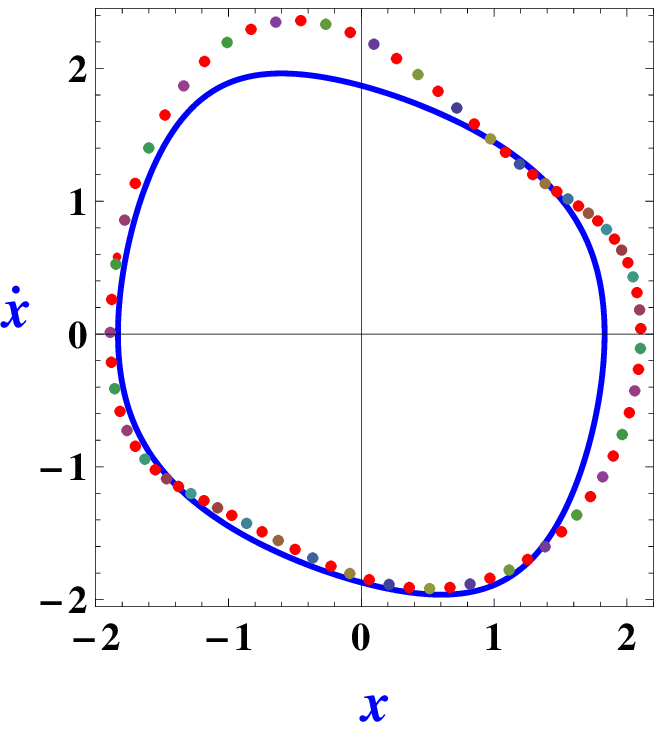}\label{f:3a}}\qquad
\qquad
\subfigure[]{\includegraphics[height=0.3\linewidth]{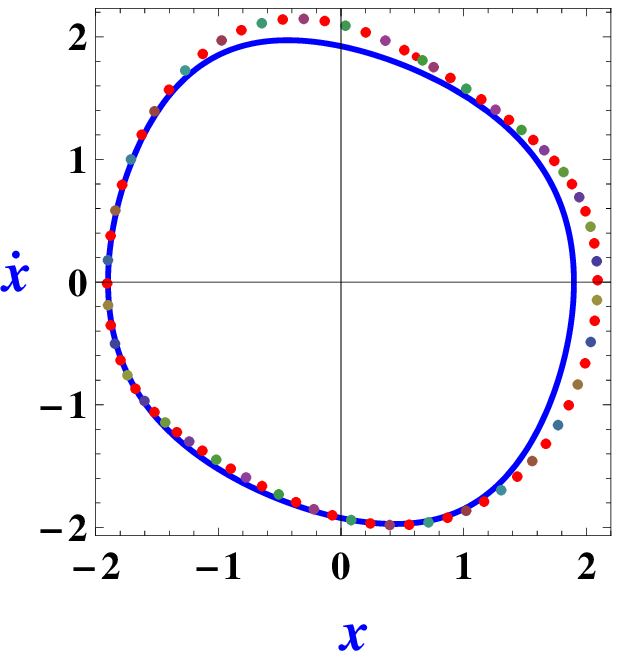}\label{f:3b}}
\end{center}
\caption{{Comparison of limit cycles given by the approximate solution
}$\left(  \text{\ref{RFRL RG New Sol}}\right)  ${ using RGM (in dotted line)
with the exact numerical limit cycle (in solid line) for }$\left(  a\right)  $
$k=F=0.5,\ \varepsilon=1.3,\ \sigma=0.07$ and $\left(  b\right)
\ k=F=0.4,\ \varepsilon=0.8,\ \sigma=0.06.$}%
\label{f:3}%
\end{figure}Similarly, taking%
\begin{equation}
k=F=0.4,\ \varepsilon=0.8,\ \sigma=0.06 \label{RFRL2 Parameter Val2}%
\end{equation}
$\left(  \text{\ref{RFRL Freq-response Eq}}\right)  $ gives three values of
$\rho$ as%
\[
4.2999\times10^{-2},\ 0.81354\text{ and }1.1435.
\]
Taking the largest value $\rho=1.1435$ we get%
\[
\Delta=9.0794\times10^{-2}>0
\]
i.e., this choice of the parameter values given by $\left(
\text{\ref{RFRL2 Parameter Val2}}\right)  $ satisfy the convergence criteria.
The limit cycle from RG solution $\left(  \text{\ref{RFRL RG New Sol}}\right)
$ is represented by dotted line along with the numerically computed (exact)
limit cycle represented by solid line in the Figure \ref{f:3b} with the values
of the parameters given by $\left(  \text{\ref{RFRL2 Parameter Val2}}\right)
$.

It follows that homotopy limit $\mu\rightarrow1^{-}$ reproduces exactly the
formal RGM solution as detailed in Appendix. As commented already, this
demonstrates the strength of classical RGM in generating approximate periodic
responses of a nonlinear system when zeroth order term happens to be
\emph{harmonic}, that agree qualitatively with the actual system response.
However, it should be clear that there indeed is a substantial mismatch (of
the order of $\sim25\%$ or more) between the exact and the theoretically
computed orbits. In the next section, we present new asymptotic results
leading to \emph{quantitatively efficient} approximations, thus improving the
conventional setting of RGM to a more efficient theory in estimating nonlinear
periodic solutions.

\section{Asymptotic Duality: IRGM and Efficient
Computation\label{Sec Asymptotic Duality}}

By asymptotic duality \cite{ds15, dss18, dsr20}, we mean an extended real
analytic framework which would support identical asymptotic (limiting) $SL(2,%
\mathbb{R}
)$ invariant (viz., translation and inversion invariant) variables, that are
allowed to vanish much more slowly compared to the original asymptotic
variable $t$ or $t^{-1}$, as the case may be, in either of the limiting
neighbourhoods $t\rightarrow0^{+}$ or $t\rightarrow\infty$ in a nonlinear system.

To be precise, a nonlinear system is generally known to enjoy a cascade of
slow scales $\tau_{n}=\varepsilon^{n}t$, so that $t$ goes to $\infty$ through
scales $\varepsilon^{-n},\ \varepsilon<1$ thus activating subdominant slower
scales $\tau_{n}\sim O(1)$ successively and hence leading to more and more
complex evolutionary structures in the system. For a strongly nonlinear
situation the relevant scales would naturally be of the form $\tau
_{n}=t/\varepsilon^{n},\ \varepsilon\geq1$ instead. Conventional asymptotic
analysis such as multiple scale method, RGM etc. are known to exploit or
invoke naturally such scales to convert the naive perturbation series into
more effective schemes capturing non-perturbative features of the nonlinear dynamics.

However, as it becomes evident from the analysis presented in Section 2,
although RGM is generally successful in capturing \emph{qualitative} features
of period orbits, fails indeed to yield \emph{quantitatively efficient}
results in strongly nonlinear problems, because of significant difference
between asymptotically approximate solutions with the exact numerical orbits.
Analogous limitations of multiple scale analysis can also be inferred. We note
that higher order calculations based on conventional treatments can not
improve the situation in a significant way, not to mention the complexity, in
under taking such computations. To eliminate this limitations in conventional
approaches, particularly in RGM, we propose the following \cite{dss18, dsr20}:

In the context of a nonlinear system the conventional \emph{linear flow} of
time as designated above by the scales of the form $\varepsilon^{-n}$ has the
potential to carry slowly varying nonlinear structures given by a
\emph{renormalized, deformed scale} of the form $t\sim\varepsilon^{-n}%
\cdot\varepsilon^{-nv(\eta)}$, where $v(\eta)\sim O(1)$ is a $SL(2,%
\mathbb{R}
)$ invariant object and $\eta\sim O(1)$ is a slow variable. To explain it more
precisely and to see what is happening let us proceed in steps:

\subsubsection*{Step 1.}

To introduce new asymptotic structure, let us begin by considering the set of
real null and divergent (either to $\infty$ or $-\infty$) sequences, denoted
$\mathcal{N}$, while $t$ being a representative divergent ( $t^{-1}$ null)
sequence. \vspace{0.2cm}

\subsubsection*{Step 2. Definitions \cite{dss18, dsr20}}

(a) The dominant characteristic scale of the nonlinear system concerned,
denoted $\varepsilon$ and assumed $\varepsilon<1$, for instance, defines a
null sequence $\varepsilon^{n}$, called the (primary) scale, for subsequent
analysis. Relative to this scale, the asymptotic variable $t\rightarrow\infty$
gets deformed images of the form
\begin{equation}
\label{pair1}t_{\pm}=\lambda_{\pm}\varepsilon^{-n}\varepsilon^{\pm nv_{\pm
}(\eta)}%
\end{equation}
satisfying the inversion rule
\begin{equation}
\label{invers}t_{+}\cdot t_{-}=\lambda(\varepsilon^{-n})^{2}.
\end{equation}
The scaling exponents $v_{\pm}(\eta)$ are, as yet unspecified real valued
function of the slow variables $\tau_{n}=\varepsilon^{n} t$, that we denote,
for brevity, by a general small scale $O(1)$ variable $\eta$ varying in
$0<<\eta<1$ or $1<\eta<<2$. It should also become clear, from above
definitions, the exponents $v_{\pm}(\eta)$ depend on the choice of the chosen
scale, this will be clarified further later. Further, out of the two variables
$t_{\pm}$, the smaller variable $t_{-}$ being an asymptotic one, $v(\eta)$
must respect the constraint $0<v_{-}(\eta)<1$. The proportionality parameters
$\lambda_{\pm}$ and $\lambda$ are generally slowly varying functions of
$\varepsilon$. Further characterization is given below. \vspace{.1cm}

(b) By definition,
\begin{equation}
v_{\pm}(\eta)=\left\vert \frac{\log t_{\pm}/\varepsilon^{-n}}{\log
\varepsilon^{n}}\right\vert +o((\log\varepsilon^{-n})^{-1}) \label{norm}%
\end{equation}
in the limit $n\rightarrow\infty$, when $\lambda_{\pm}=\mathrm{real\ const.}%
+o(1)$ or $o(1)$. Under this constraint, $\lambda$ itself may be slowly
varying constant or vanishing. Consequently, it follows from inversion rule
(\ref{invers}) that
\begin{equation}
v_{+}(\eta)=\ v_{-}(\eta) \label{selfdual}%
\end{equation}
hence, $v_{+}(\eta)$ must also respect the constraint $0<v_{+}(\eta)<1$.
Consequently, the scaling exponents are said to be \emph{self-dual} exponents
and denoted by the symbol $v(\eta)\in(0,1)$. The reason for this nomenclature
is explained by the following Proposition.

\begin{proposition}
[\textbf{Duality or }$SL(2,%
\mathbb{R}
)$\textbf{ Invariance Properties}]\cite{dss18, dsr20}\label{invprop}
\label{invariance} The scaling exponents $v_{\pm}(\eta)$, denoted for brevity
by $v(\eta)$, satisfy the following invariance properties. \newline(1)
$v(k\eta)=v(\eta)$, for any constant $k$, \newline(2) $v(\eta^{-1})=v(\eta)$,
\newline(3) $v(\eta-\eta_{0})=v(\eta)$, for any $\eta_{0}$ satisfying, for
instance, the inequality $0<|\eta_{0}|\leq|\eta|\leq1$, which are all valid up
to an additive null sequence $o(1)$.
\end{proposition}

\begin{proof}
The proof of these invariance properties, viz., (1) Scaling, (2) Inversion and
(3) Translation, follow directly from Definitions (a) and (b). However, for
completeness, we sketch here the proofs of inversion and translation
invariance, (2) and (3) respectively.

\emph{Proof of (2):} For definiteness, we work with $t_{+}:=\tau^{-1}$ (say).
Then $\tau<<1$, and the corresponding scale is $\varepsilon^{n}$. By
definition (\ref{norm}),%
\begin{equation}
v(\eta^{-1})=\left\vert \frac{\log\tau/\varepsilon^{n}}{\log\varepsilon^{n}%
}\right\vert +o(1)=\left\vert \frac{\log t_{+}/\varepsilon^{-n}}%
{\log\varepsilon^{n}}\right\vert +o(1)=v(\eta),\label{norm2}%
\end{equation}
where $\eta^{-1}\lessapprox1$ when $\eta\gtrapprox1$, proving the inversion
invariance. (Q.E.D.)

\emph{Proof of (3):} Again, by definition (\ref{norm})%
\begin{equation}
v_{\pm}(\eta-\eta_{0})=\left\vert \frac{\log\varepsilon^{n}(t_{+}-t_{+0}%
)}{\log\varepsilon^{n}}\right\vert +o(1)=\left\vert \frac{\log\varepsilon
^{n}t_{+}}{\log\varepsilon^{n}}\right\vert +\left\vert \frac{\log
(1-t_{+0}/t_{+})}{\log\varepsilon^{n}}\right\vert +o(1)=v(\eta)\label{stineq}%
\end{equation}
since the second logarithmic term in second equality vanishes faster ( here
$0<<t_{+0}<t_{+}$), proving the translation invariance. (Q.E.D.)
\end{proof}

The translation invariance also reveals the ultrametric property of the
scaling exponent $v(\eta)$. Indeed $v(\eta)$ defines an ultrametric norm on
the set of the asymptotic variables $\mathcal{N}$ satisfying the strong
triangular property: $v(\eta+\eta_{0})\leq\max\{v(\eta),v(\eta_{0})\}$, $\eta
$, $\eta_{0}$ satisfying above inequalities. For a detail proof see
\cite{dss18,dsr20}.

Let us remark that the mapping $t\mapsto v(\eta)$ can be interpreted as a
renormalization group action, assigning a finite slowly varying value
$v(\eta)$ to a diverging or null sequence (i.e. $t$ or $t^{-1}$). The scaling
exponent $v(\eta)$ could therefore be referred either as the renormalized
scaling exponent or the deformation scaling factor. \vspace{.2cm}

\subsubsection*{Step 3. Self-dual vis-$\acute{a}$-vis Trivial Ultrametric
\vspace{0.1cm}}

The stronger triangular inequality satisfied by the self dual exponent
$v(\eta)$ tells that $v(\eta)$ acts on $\mathcal{N}$ as an ultrametric norm.
Under such ultrametric norm the asymptotic set $\mathcal{N}$ naturally has the
structure of a totally disconnected set, while the value set of the norm
$v(\eta)$ can at most be countable \cite{nonarc}. The particular norm that
specifies a constant $O(1)$ value to each element of $\mathcal{N}$ is the
trivial ultrametric. It also follows in the present context that this trivial
norm, by its very definition, is self dual. This ultrametric norm is constant
over the set $\mathcal{N}$, but could indeed be a function on slow variables
characteristic of the given nonlinear system. Consequently, the self dual norm
$v(\eta)$ is a trivial ultrametric norm over $\mathcal{N}$ when the
proportionality parameter $\lambda$ in the inversion rule (\ref{invers}) is a
slowly varying constant.

Incidentally, we remark that depending on the choice of a scale $\delta$
(say), the self dual exponent $v(\eta)$ could be extended to more general
duality relations, viz., (1) weakly self dual (or simply dual) exponents
$v_{+}(\eta)=\alpha(\delta)v_{-}(\eta),\ \alpha(\delta)>0$, when $\lambda$ in
(\ref{invers}) is slowly varying as $O((\log\delta)^{k})$, $k$ a constant, and
strictly dual when $v_{+}(\eta)v_{-}(\eta)=\mu(\delta),\mu>0$ and $\lambda$ in
(\ref{invers}) is slowly varying as $O(\delta^{\kappa})$, $\kappa$ a constant
\cite{dsr20}. Parameters such $\alpha$, $\mu$ etc. appearing in above
definitions can also be slowly varying following various possible patterns
that could be discerned in a specific nonlinear problem. Possible applications
of weakly dual and strictly dual cases would be considered in the context of
period doubling bifurcations in nonlinear oscillations in subsequent works.
\vspace{.2cm}

\subsubsection*{Step 4. Analytic Function Space \vspace{0.1cm}}

The space of periodic solutions of nonlinear oscillations consists, in
general, of analytic functions. One needs to specify the action of the
ultrametric scaling exponents on the space of analytic functions. This would
later facilitate one to improve upon the RGM approximate solutions of limit
cycle solutions to an efficient computation.

Let $(x(t),y(t)),\ y(t)=\dot{x}(t)$ be the exact phase space solution of the
nonlinear periodic orbit and $(x_{0}(t),y_{0}(t))$ be the corresponding RGM
approximation. Recalling from Section 2 the fact that an approximate periodic
solution is retrieved from naive perturbation series following RG algorithms
in the limit $t\rightarrow\infty$ through linear scales of $\varepsilon^{-n}$,
when amplitude and phase flow equations attain stationary condition. Let
$t_{0}=k\varepsilon^{-n}$ be a sufficiently large time instant. In a numerical
computation one usually replaces $\infty$ by $t_{0}$ for $n$ sufficiently large.

In the context of the so-called IRGM, we introduce nonlinear deformation in
the above linear flow of time through self dual scaling exponent $v(\eta)$ and
replace the large time instance $t_{0}$ instead by $T_{0}=t_{0}\cdot
t_{0}^{-v(\eta)}$. For a sufficiently small, but nevertheless $O(1)$,
$v(\eta)$, the linearized deformed time has the form $T_{0}=t_{0}-v(\eta
)t_{0}\log t_{0}:=t_{0}+\chi(\eta)$.

Now, there is a caveat here. We make an \emph{important} comment on the nature
of \emph{time scale deformations} in the context of a dynamical system.
\vspace{.2cm}

\subsubsection*{Step 4.1. Dynamical time scale \vspace{0.1cm}}

A nonlinear system generally involves multiple \emph{characteristic} scales. A
periodic oscillation in a slow-fast system such as singularly perturbed
Rayleigh equation is characterized by the presence of vary fast and slow rates
of motion along two independent phase space directions
\cite{jordan_nonlinear_1999, Xu-Luo-2019}. Consequently, the deformation
factor, considered for the time scale $t_{0}$, viz., $t_{0}^{-v(\eta)}$
involving only one deformation exponent $v(\eta)$ is too restrictive, and
frozen, so to speak, in character. To allow for a more flexible and dynamic
character accommodating different rates of evolution, one may introduce, in
the following, more general \emph{dynamical deformation scales} given by
$\sigma(t_{0}):=t_{0}^{-1/\psi(t_{0})}$, so that the {scale dependent} $SL(2,%
\mathbb{R}
)$ deformation exponent is now given by
\begin{equation}
v_{\psi}(\eta)=\left\vert \frac{\log T_{0}/t_{0}}{\log\sigma}\right\vert +O(1)
\label{dynnorm}%
\end{equation}
where $\psi(t_{0})\geq0$ is, in general, an analytic function, depending only
on the phase space variables $x,\ y,\ldots$. The deformation scale is dynamic,
as the factor $\sigma$ accommodates naturally any critical behaviour that a
system may experience, by dynamically adjusting or expanding the original
scale $t_{0}$ to a more larger scale, specific to the system behaviour, close
to a critical point (set) satisfying $\dot{x}=0$, or $\dot{y}=0$ etc. In the
present nonlinear oscillation problem (slow-fast system), as detailed in the
following, we choose to work with $\psi_{x}:=k_{x}^{-1}|\dot{x}_{0}|t_{0}\log
t_{0}$ and $\psi_{y}:=k_{y}^{-1}|\dot{y}_{0}|t_{0}\log t_{0}$, for some
positive constants $k_{i},\ i=x,y$, where $x_{0}$ and $y_{0}$ are the RGM
approximated solutions of $x$ and $y$.

The above choice is guided by the above mentioned property of slow fast
systems involving different rates of oscillatory motion. As it turns out, such
choices along $x$ and $y$ directions could successfully yield very high level
of precision computations, exceeding $98\%$ or above accuracy level. More
interestingly, one expects a sort of visual (graphical) verification of slow
fast action of the the dynamic scales in the form of a \emph{condensation and
rarefaction phenomenon} in the distribution of data points along the fast and
slow phase paths respectively, in a numerical application of the method in a
relaxation oscillation. In the next section, we indeed get explicit
verifications of this sort of condensation and rarefaction phenomenon in the
phase diagrams of periodic orbits. One notes also that the corresponding
nonlinear, dynamically deformed time, denoted as $T_{0}^{\psi}$, has the form
\begin{equation}
T_{0}^{\psi}=t_{0}-v_{\psi}(\eta)t_{0}\log\sigma(t_{0}): =t_{0}+\frac{v_{\psi
}(\eta)}{\psi(t_{0})}t_{0}\log t_{0}:=t_{0}+\chi_{\psi}(t_{0},\eta).
\end{equation}

Next, turning back to our discussion of analytic functions, we replace
approximate solution $(x_{0}(t),y_{0}(t))$ by corresponding dynamically
deformed solution $\tilde x_{0}=x_{0}(t_{0}+\chi_{x}(t_{0},\eta))$ and $\tilde
y_{0}=y_{0}(t_{0}+\chi_{y}(t_{0},\eta))$, where $\chi_{i}(t_{0},\eta
)=\frac{v_{\psi}(\eta)}{\psi_{i}(t_{0})}t_{0}\log t_{0}, \ i=x,y$. Assuming
analyticity of the solutions one then gets the linearized solutions as%

\[
\tilde x_{0}=x_{0}(t_{0}) +\dot x_{0}(t_{0})\chi_{x}(t_{0},\eta), \ \tilde
y_{0}=y_{0}(t_{0}) +\dot y_{0}(t_{0})\chi_{y}(t_{0},\eta).
\]

The precise mismatch between the exact and approximated solutions in IRGM is
now estimated as%

\begin{align}
x(t_{0})-x_{0}(t_{0})  &  =\dot{x}_{0}(t_{0})\chi_{x}(\eta)=k_{x}v_{x}%
(\eta),\label{error}\\
y(t_{0})-y_{0}(t_{0})  &  =\dot{y}_{0}(t_{0})\chi_{y}(\eta)=k_{y}v_{y}(\eta)
\label{error in y}%
\end{align}
when we fix the dynamical scales along $x$ and $y$ directions as $\psi
_{x}:=k_{x}^{-1}|\dot{x}_{0}|t_{0}\log t_{0}$ and $\psi_{y}:=k_{y}^{-1}%
|\dot{y}_{0}|t_{0}\log t_{0}$ respectively, for some positive adjustable
constants $k_{x}$ and $k_{y}$.

We next invoke analyticity of the scaling exponents $v_{x}$ and $v_{y}$ and
then may approximate $v_{i}(\eta), \ i=x,y$ by polynomials $v_{i}(\eta
)=\sum_{1}^{m}a_{ij}\eta^{j}, \ i=x,y$. One then specifies the constants
$a_{ij}$ and free parameters $k_{i}$ so as to obtain best fit approximations
of $v_{i}(\eta)$, leading to high precision periodic solution by IRGM.
However, as it will become evident, single piece polynomial approximation may
not be sufficient to yield efficient matching of orbits. One indeed requires
multi-component polynomial functions for $v_{i}$s over a complete period of
oscillation, to yield an efficient computation. Since we are presenting our
computational results using the state of art computational platform of
Mathematica, the question of assuring smoothness at the boundary (branch)
points where two best fitted polynomials are supposed to match smoothly, so as
to yield a smooth function on the entire period of the periodic oscillation
concerned, could not be tackled in the present paper. To address this problem
fully, one perhaps needs to invoke independent code, that we leave for future
work, though, theoretically, determining smoothly matched best fit polynomial
curves across boundary points should not be difficult with so many available
adjustable parameters. We close this section with the following observation.

The multi-component polynomial best fit estimations of correction terms
$V_{i}(\eta)$ in (\ref{error}) and (\ref{error in y}) correspond to best
possible piece-wise analytic approximations of the residual periodic functions
$x(t_{0})-x_{0}(t_{0})$ and $y(t_{0})-y_{0}(t_{0})$ respectively, that arise
from inaccurate approximations of exact solutions in the conventional
asymptotic techniques. The identification and efficient estimations of these
residual inaccuracies may be considered as significant advances in the context
of the present formalism.

\section{Examples: Efficient estimation of relaxation
cycles\label{Sec Examples}}

To establish and demonstrate the significance of IRGM accommodating novel
asymptotic structures, we now present the main steps of numerical
calculations, with corresponding graphical representations, for Example (A)
singularly perturbed Rayleigh equation $\left(  \text{\ref{SFRL Eqn}}\right)
$ and Example (B) ( regular) Rayleigh equation $\left(  \text{\ref{RFRL Eqn}%
}\right)  $ with a large nonlinearity parameter $\varepsilon>1$. We refer to
Section \ref{SRLE} and \ref{Reg RLE} for the homotopy aided RGM computed
solutions of these two cases.

\subsubsection*{\textbf{Example A}}

The RG solution after putting $\mu=1$ for the singular equation $\left(
\text{\ref{SFRL Eqn}}\right)  $ is given by $\left(  \text{\ref{SRLE RG Sol}%
}\right)  $. In Figure \ref{f:1b} we see that the exact limit cycle and its
approximation by RGM are significantly different (i.e. accuracy not exceeding
$75\%$) for $F=0.4,$\ $\omega=2\ \left(  \text{or }\varepsilon=0.25\right)  $,
$\sigma=0.01\omega^{2}$.

To improve upon the above limitations, we first fix slower $O(1)$ variable
$\eta$ that is supposed to appear in the scaling exponents of (\ref{error}).
For $\varepsilon=0.25$, by observation, one finds a complete oscillation of
the relaxation cycle concerned in the time interval $190<t<200$, and the
\emph{minimum} value of $n$ so that $\eta=\varepsilon^{n}t$ falls in $(0,1)$
is $n=4$. Indeed, one verifies that for complete period $190<t<200$, the slow
variable $\eta$ satisfies $0.7421875<\eta<0.78125$ for $n=4$.

Next, one notes a phase difference between the numerically computed solution
$x(t)$ and the RGM solution $x_{0}(t)$ in the said interval (Figure
\ref{f:4a}). \begin{figure}[ptb]
\begin{center}
\centering%
\begin{tabular}
[c]{rcr}%
\subfigure[]{\includegraphics[height=0.3\linewidth,width=0.41\linewidth]{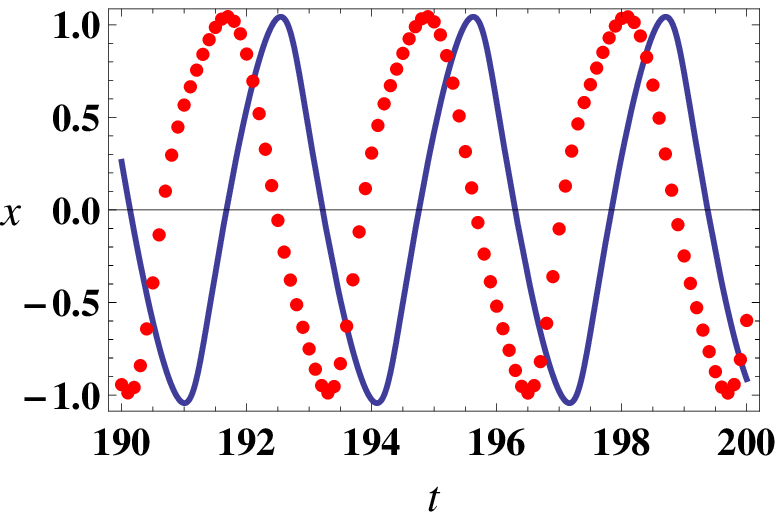}\label{f:4a}} &
\qquad\qquad &
\subfigure[]{\includegraphics[height=0.3\linewidth,width=0.42\linewidth]{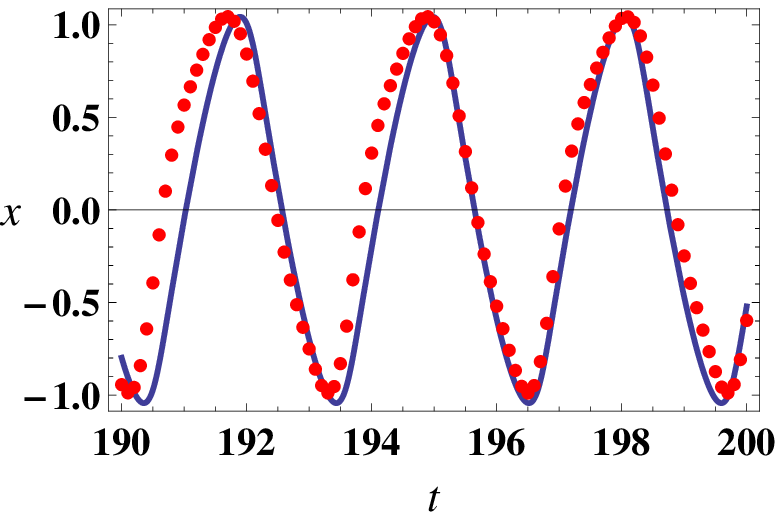}\label{f:4b}}\\
\subfigure[]{\includegraphics[height=0.3\linewidth]{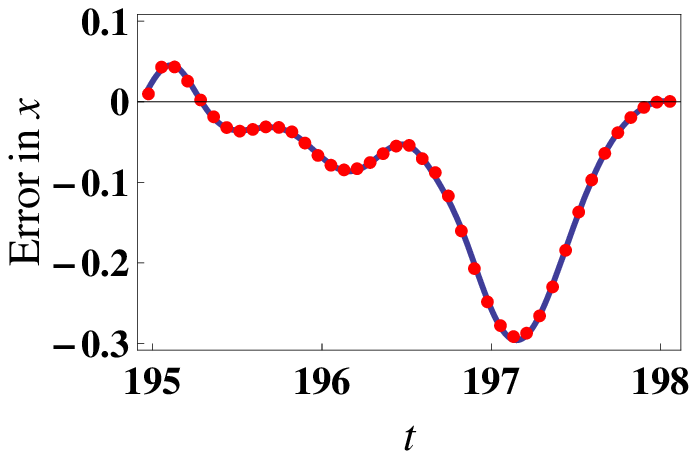}\label{f:4c}} &
\qquad\qquad &
\subfigure[]{\includegraphics[height=0.3\linewidth]{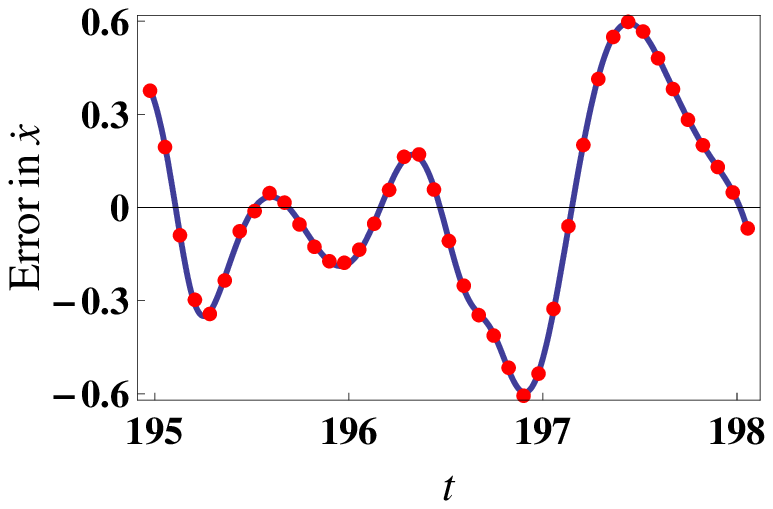}\label{f:4d}}\\
\subfigure[]{\includegraphics[height=0.3\linewidth,width=0.43\linewidth]{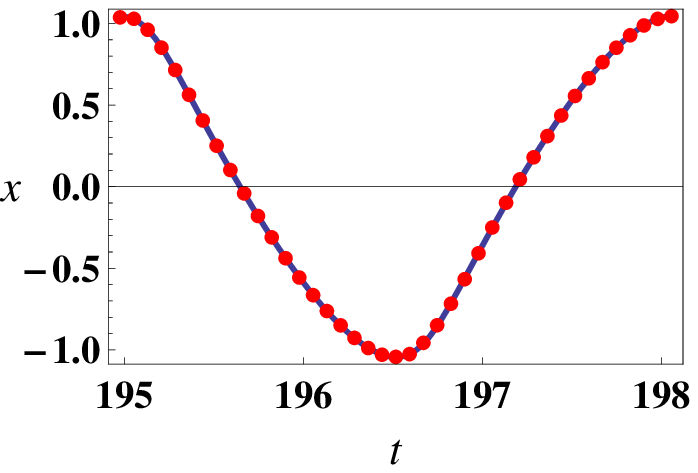}\label{f:4e}} &
\qquad\qquad &
\subfigure[]{\includegraphics[height=0.3\linewidth,width=0.42\linewidth]{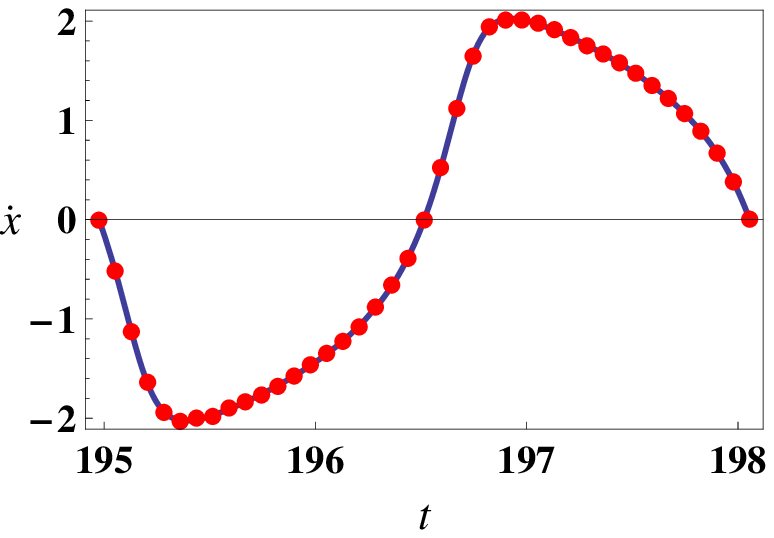}\label{f:4f}}
\end{tabular}
\end{center}
\caption{Improving the approximation shown in Figure \ref{f:1b} discussed in
Section \ref{SRLE}. Here \textbf{Solid} lines represent exact solutions and
\textbf{Dotted} lines represent approximate solutions.\medskip\newline$\left(
a\right)  $ Comparison of the exact solution $x\left(  t\right)  $ of $\left(
\text{\ref{SFRL Eqn}}\right)  $ and the RGM approximation given by $\left(
\text{\ref{SRLE RG Sol}}\right)  $ for $190<t<200.$ $\left(  b\right)  $
Comparison of $x\left(  t+0.65\right)  $ and the RGM approximation given by
$\left(  \text{\ref{SRLE RG Sol}}\right)  $ showing better agreement for
$190<t<200$. $\left(  c\right)  $ Comparison of exact error $x\left(
t+0.65\right)  -x_{0}\left(  t\right)  $ and its estimate of $\psi_{x}\left(
t\right)  $ given by $\left(  \text{\ref{Error1 Estimate in x}}\right)  $ in
$\left[  r_{1},r_{3}\right]  $. $\left(  d\right)  $ Comparison of exact error
$\dot{x}\left(  t+0.65\right)  -\dot{x}_{0}\left(  t\right)  $ and its
estimate of $\psi_{y}\left(  t\right)  $ given by $\left(
\text{\ref{Error1 Estimate in y}}\right)  $ in $\left[  r_{1},r_{3}\right]  $.
$\left(  e\right)  $ Comparison of $x\left(  t+0.65\right)  $ and its estimate
$x_{0}\left(  t\right)  +\psi_{x}\left(  t\right)  $. $\left(  f\right)  $
Comparison of $\dot{x}\left(  t+0.65\right)  $ and its estimate $\dot{x}%
_{0}\left(  t\right)  +\psi_{y}\left(  t\right)  $.}%
\label{f:4}%
\end{figure}To compensate the phase difference, that is essential for
determination of numerical accuracy of the technique, we verify that a
suitable time translated $x(t)\rightarrow x(t+0.65)$ have approximately
identical phase with $x_{0}(t)$, viz., (Figure \ref{f:4b})
\[
x\left(  t+0.65\right)  \approx x_{0}\left(  t\right)  ,
\]
in\ phase. It turns out that this time translation along $x$ direction is
inherited by $y$ curves, viz., $y(t+0.65)$ and $y_{0}(t)$, both having almost
identical phases. This suffices us to replace $\left(  \text{\ref{error}%
}\right)  $ and $\left(  \text{\ref{error in y}}\right)  $ by new phase
shifted error equations
\begin{align}
x(t+0.65)-x_{0}(t)  &  =\dot{x}_{0}(t)\chi_{x}(\eta)=k_{x}v_{x}(\eta
)\equiv\psi_{x}(t),\label{error1}\\
y(t+0.65)-y_{0}(t)  &  =\dot{y}_{0}(t)\chi_{y}(\eta)=k_{y}v_{y}(\eta
)\equiv\psi_{y}(t), \label{error2 in y}%
\end{align}
where $\eta=\varepsilon^{n}t$ for $190<t<200,\ \varepsilon=0.25,\ n=4$ using
best fit polynomial curves as follows.

In order to identify the period of the oscillation, we find that $x\left(
t+0.65\right)  $ $\left(  \text{not }x_{0}\left(  t\right)  \right)  $ has
successive maximum at the points $r_{1}=194.975,\ r_{3}=198.055$ and it has a
minimum in between them at $r_{2}=196.515$ in the interval $190<t<200$. Thus,
the period of the function $x\left(  t+0.65\right)  $ is $T=r_{3}%
-r_{1}=198.055-194.975=3.08\text{. }$The graph of $x\left(  t+0.65\right)
-x_{0}\left(  t\right)  $ is shown in Figure \ref{f:4c} on $\left[
r_{1},r_{3}\right]  $ through solid line. It is difficult to estimate the
above function using a single polynomial. So, by careful observation we divide
the interval $\left[  r_{1},r_{3}\right]  $ into subintervals
\[
\left[  r_{1},195.7\right]  ,\ \left[  195.7,196.6\right]  ,\ \left[
196.6,197.6\right]  ,\ \left[  197.6,r_{3}\right]
\]
and then estimate them using $100$ equidistant data points from each
subintervals through Mathematica by the piecewise function (expected to be
continuous and smooth, but yet to be verified (c.f. Section 3, item 4.1,
remarks on error estimation))\begingroup\scalefont{1.0}%

\begin{equation}
\psi_{x}\left(  t\right)  =\left\{
\begin{array}
[c]{ll}%
-6.08565\times10^{9}+1.2456\times10^{8}t-956058.t^{2}+3261.41t^{3}%
-4.17212t^{4}, & r_{1}\leq t\leq195.7\\
-3.44505\times10^{9}+7.02609\times10^{7}t-537357.t^{2}+1826.54t^{3}%
-2.32822t^{4}, & 195.7<t\leq196.6\\
-4.11615\times10^{9}+8.35121\times10^{7}t-635386.t^{2}+2148.53t^{3}%
-2.72444t^{4}, & 196.6<t\leq197.6\\
-563206+5648.14t+0.153819t^{2}-0.143543t^{3}+0.000359855t^{4}, & 197.6<t\leq
r_{3}%
\end{array}
\right.  \label{Error1 Estimate in x}%
\end{equation}
\endgroup the graph of which is shown in the Figure \ref{f:4c} through dotted
line. Similarly,
\begingroup\scalefont{0.8}%
\begin{equation}
\psi_{y}\left(  t\right)  =\left\{
\begin{array}
[c]{ll}%
-8.09646\times10^{10}+1.24433\times10^{9}t-4.25298\times10^{6}t^{2}%
-21721.8t^{3}+167.046t^{4}-0.285255t^{5}, & r_{1}\leq t\leq195.5\\
-9.50313\times10^{9}+1.45469\times10^{8}t-494748.t^{2}-2526.17t^{3}%
+19.3312t^{4}-0.0328788t^{5}, & 195.5<t\leq196.3\\
-4.16274\times10^{10}+6.35172\times10^{8}t-2.15314\times10^{6}t^{2}%
-10963.2t^{3}+83.6184t^{4}-0.141766t^{5}, & 196.3<t\leq196.85\\
1.5638\times10^{10}-2.37679\times10^{8}t+802502.t^{2}+4071.01t^{3}%
-30.9271t^{4}+0.0522281t^{5}, & 196.85<t\leq197.7\\
3.24289\times10^{7}-327797.t+0.102438t^{2}+8.37216t^{3}-0.0211569t^{4}, &
197.7<t\leq r_{3}%
\end{array}
\right.  \label{Error1 Estimate in y}%
\end{equation}%
\endgroup
the graph of which is compared with the exact error $\dot{x}\left(
t+0.65\right)  -\dot{x}_{0}\left(  t\right)  $ in the Figure \ref{f:4d}. The
graphs of $x\left(  t+0.65\right)  $, $\dot{x}\left(  t+0.65\right)  $ and
their respective estimates $x_{0}\left(  t\right)  +\psi_{x}\left(  t\right)
$, $\dot{x}_{0}\left(  t\right)  +\psi_{y}\left(  t\right)  $ are shown in
Figure \ref{f:4e} and Figure \ref{f:4f} respectively.

To estimate the percentage error in above numerical calculations we cannot use
the standard error formula, for instance, in $x$ direction
\[
\max\frac{x\left(  t+0.65\right)  -\left(  x_{0}\left(  t\right)  +\psi
_{x}\left(  t\right)  \right)  }{x\left(  t+0.65\right)  }\times100\text{ for
}r_{1}\leq t\leq r_{3}%
\]
for the fact that $x\left(  t+0.65\right)  $ may vanish in $\left[
r_{1},r_{3}\right]  $. So we estimate the error with respect to the amplitude
$\tilde{a}$ of $x$ using the formula%
\begin{equation}
E_{x}=\max\frac{x\left(  t+0.65\right)  -\left(  x_{0}\left(  t\right)
+\psi_{x}\left(  t\right)  \right)  }{\tilde{a}}\times100\text{ for }r_{1}\leq
t\leq r_{3}. \label{effx}%
\end{equation}
We find that%
\[
\tilde{a}=\max_{r_{1}\leq t\leq r_{3}}\left\vert x\left(  t+0.65\right)
\right\vert =1.04367\text{ and }\max_{r_{1}\leq t\leq r_{3}}\left\{  x\left(
t+0.65\right)  -\left(  x_{0}\left(  t\right)  +\psi_{x}\left(  t\right)
\right)  \right\}  =0.0127652
\]
so that%
\[
E_{x}=\frac{0.0127652}{1.04367}\times100=1.22\%
\]
and hence we achieve an accuracy of $\left(  100-E_{x}=\right)  \ 98.78\%$ in
this estimation.

Similarly, for estimation of $y\left(  t+0.65\right)  =\dot{x}\left(
t+0.65\right)  $ by $\dot{x}_{0}\left(  t\right)  +\psi_{y}\left(  t\right)  $
we use the formula for percentage error as%
\begin{equation}
E_{y}=\max\frac{\dot{x}\left(  t+0.65\right)  -\left(  \dot{x}_{0}\left(
t\right)  +\psi_{y}\left(  t\right)  \right)  }{\tilde{b}}\times100\text{ for
}r_{1}\leq t\leq r_{3}, \label{effy}%
\end{equation}
where $\tilde{b}$ is the amplitude of $\dot{x}\left(  t+0.65\right)  $ in
$\left[  r_{1},r_{3}\right]  $ i.e.,
\[
\tilde{b}=\max_{r_{1}\leq t\leq r_{3}}\left\vert \dot{x}\left(  t+0.65\right)
\right\vert \text{.}%
\]
Here,%
\[
\tilde{b}=2.02533\text{ and }\max_{r_{1}\leq t\leq r_{3}}\left\{  \dot
{x}\left(  t+0.65\right)  -\left(  \dot{x}_{0}\left(  t\right)  +\psi
_{y}\left(  t\right)  \right)  \right\}  =0.0311104
\]
so that%
\[
E_{y}=\frac{0.0311104}{2.02533}\times100=1.54\%
\]
and hence we achieve an accuracy of $\left(  100-E_{y}=\right)  \ 98.46\%$ in
this estimation. The estimated limit cycle with $x_{0}\left(  t\right)
+\psi_{x}\left(  t\right)  $ and $\dot{x}_{0}\left(  t\right)  +\psi
_{y}\left(  t\right)  $ is compared with exact limit cycle in Figure
\ref{f:5}. \begin{figure}[ptb]
\begin{center}
\includegraphics[width=0.3\linewidth,width=0.3\linewidth]{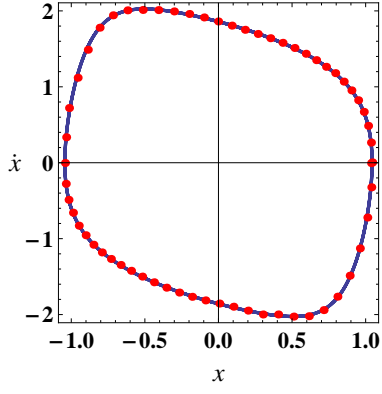}
\end{center}
\caption{The solid line represents numerically computed exact limit cycle as
discussed in Section \ref{SRLE} and the dotted line represents the estimated
limit cycle by $x_{0}\left(  t\right)  +\psi_{x}\left(  t\right)  $ and
$\dot{x}_{0}\left(  t\right)  +\psi_{y}\left(  t\right)  $.}%
\label{f:5}%
\end{figure}\vspace{0.2cm}

\subsubsection*{\textbf{Example B}}

Next, we present efficient computations in the context of the Rayleigh
Equation with large nonlinearity $\varepsilon>1$. The homotopy RG solution of
Section \ref{Reg RLE}, after putting $\mu=1$, is given by $\left(
\text{\ref{RFRL RG New Sol}}\right)  $. As usual, the exact limit cycle and
its approximation by homotopy RGM (Figure \ref{f:3a}) are significantly
different (accuracy level being less than $75\%$) for $F=0.5,$\ $\varepsilon
=1.3$, $\sigma=0.07$.

Now, following the steps of (A), for a complete cycle we choose $185<t<200$,
so that the minimum value of $n$ is now fixed at $n=21$ and the slow variable
$\eta=\varepsilon^{-21}t\in\left(  0,1\right)  $, for $\varepsilon=1.3$. In
order to remove relative phase difference between exact and approximated
solution, one time translates the exact solution $x(t)\rightarrow x(t+2.7)$,
and consider the phase corrected error equations in the form
\begin{align}
x(t+2.7)-x_{0}(t)  &  =\dot{x}_{0}(t)\chi_{x}(\eta)=k_{x}v_{x}(\eta
)\equiv\tilde{v}_{x}(\eta)=\psi_{x}(t),\label{error2}\\
y(t+2.7)-y_{0}(t)  &  =\dot{y}_{0}(t)\chi_{y}(\eta)=k_{y}v_{y}(\eta
)\equiv\tilde{v}_{y}(\eta)=\psi_{y}(t),
\end{align}
where $\eta=\varepsilon^{-21}t$ for $185<t<200,\ \varepsilon=1.3,$ using
multi-component polynomial curve fitting. As before, we estimate the period of
$x\left(  t+2.7\right)  $ as $T=r_{3}-r_{1}=5.759$, where the maxima of the
oscillation are attained at $r_{1}=190.507,\ r_{3}=196.266$, when the minimum
value is at $r_{2}=193.387$ for $185<t<200$.

Thereafter, we use here three piece polynomial fitting to annul the error
difference in (\ref{error2}) as \begingroup\scalefont{0.9}%
\begin{equation}
\psi_{x}\left(  t\right)  =\left\{
\begin{array}
[c]{ll}%
-7.29799\times10^{7}+1.51571\times10^{6}t-11804.4t^{2}+40.8578t^{3}%
-0.0530298t^{4}, & r_{1}\leq t\leq191.9\\
-1.42254\times10^{8}+2.9514\times10^{6}t-22962.5t^{2}+79.4005t^{3}%
-0.102957t^{4}, & 191.9<t\leq194\\
-2.38474\times10^{8}+4.893\times10^{6}t-37647.7t^{2}+128.74t^{3}%
-0.165089t^{4}, & 194<t\leq r_{3}.
\end{array}
\right.  \label{Error2 Estimate in x}%
\end{equation}
\endgroup Similarly, one estimates $\dot{x}\left(  t+2.7\right)  $ using
$\dot{x}_{0}\left(  t\right)  +\psi_{y}\left(  t\right)  $ by dividing the
interval $\left[  r_{1},r_{3}\right]  $ instead in four subintervals and then
estimated $\dot{x}\left(  t+2.7\right)  -\dot{x}_{0}\left(  t\right)  $ by the
piecewise function\begingroup\scalefont{0.9}%
\begin{equation}
\psi_{y}\left(  t\right)  =\left\{
\begin{array}
[c]{cl}%
-3.48561\times10^{9}+7.29408\times10^{7}t-572390.t^{2}+1996.32t^{3}%
-2.61094t^{4}, & r_{1}\leq t\leq191.6\\
-6.14729\times10^{8}+1.27786\times10^{7}t-99611.6t^{2}+345.106t^{3}%
-0.448357t^{4}, & 191.6<t\leq193.4\\
2.1483\times10^{9}-4.43429\times10^{7}t+343228.t^{2}-1180.75t^{3}%
+1.52322t^{4}, & 193.4<t\leq194.6\\
7.30449\times10^{8}-1.4952\times10^{7}t+114771.t^{2}-391.545t^{3}%
+0.500907t^{4}, & 194.6<t\leq r_{3}.
\end{array}
\right.  \label{Error2 Estimate in y}%
\end{equation}
\endgroup The graph of $x\left(  t+2.7\right)  $, $\dot{x}\left(
t+2.7\right)  $ and their respective estimates $x_{0}\left(  t\right)
+\psi_{x}\left(  t\right)  $, $\dot{x}_{0}\left(  t\right)  +\psi_{y}\left(
t\right)  $ are shown in Figure \ref{f:6}. \begin{figure}[ptb]
\begin{center}
\centering%
\begin{tabular}
[c]{rcr}%
\subfigure[]{\includegraphics[height=0.3\linewidth,width=0.38\linewidth]{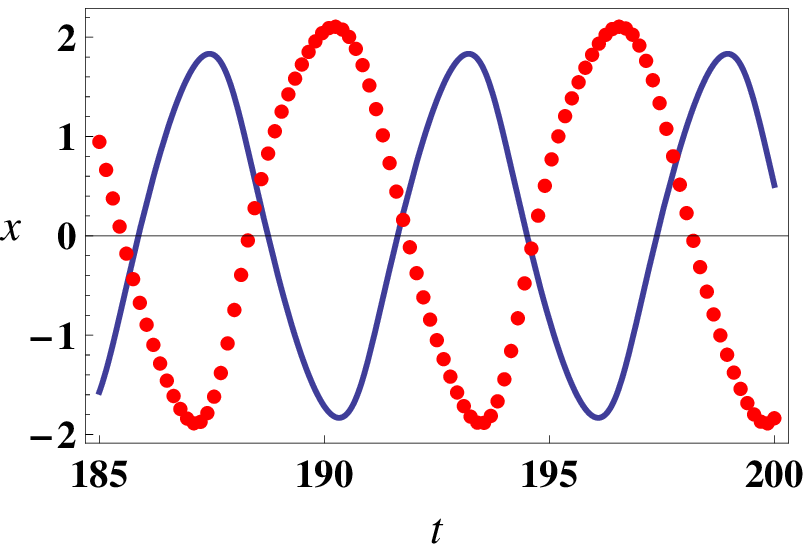}\label{f:6a}} &
\qquad\qquad &
\subfigure[]{\includegraphics[height=0.3\linewidth,width=0.36\linewidth]{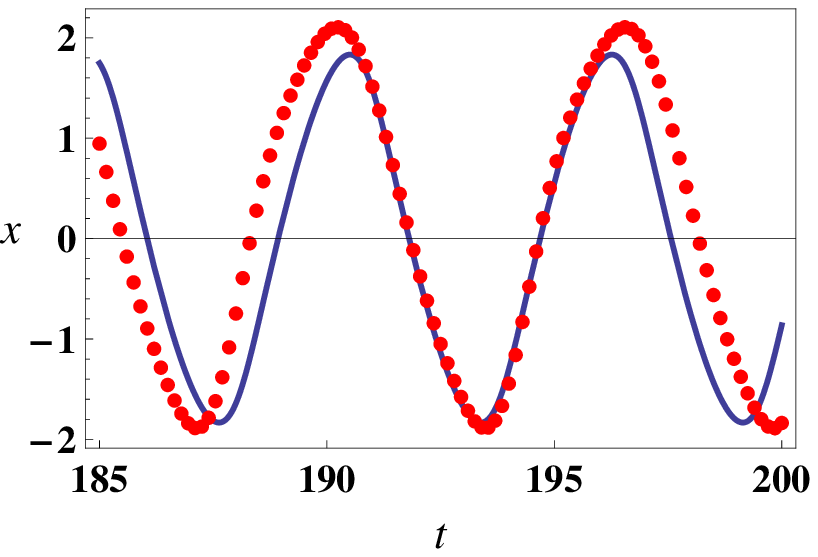}\label{f:6b}}$\left.
\ \right.  $\\
\subfigure[]{\includegraphics[height=0.3\linewidth,width=0.40\linewidth]{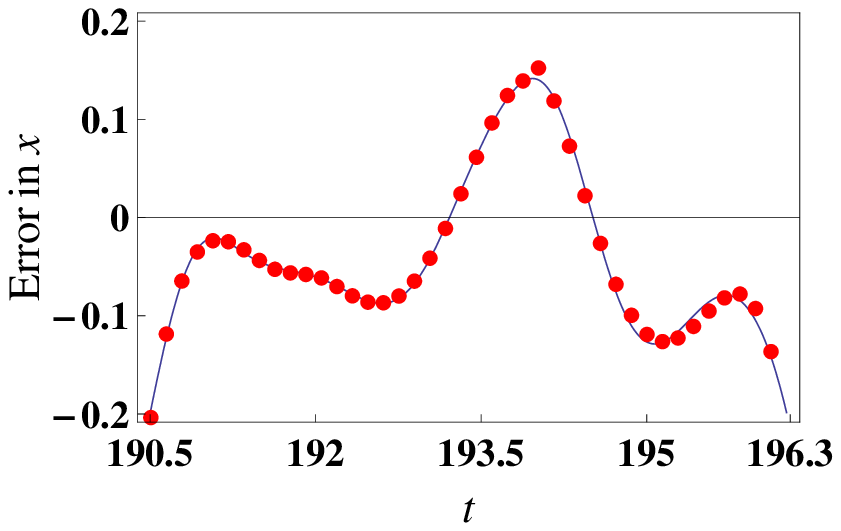}\label{f:6c}} &
\qquad\qquad &
\subfigure[]{\includegraphics[height=0.3\linewidth,width=0.40\linewidth]{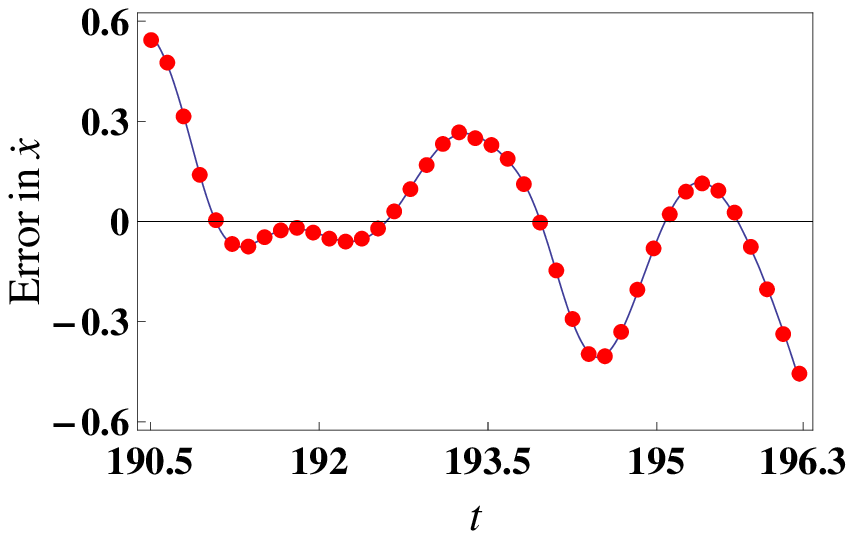}\label{f:6d}}\\
\subfigure[]{\includegraphics[height=0.3\linewidth,width=0.37\linewidth]{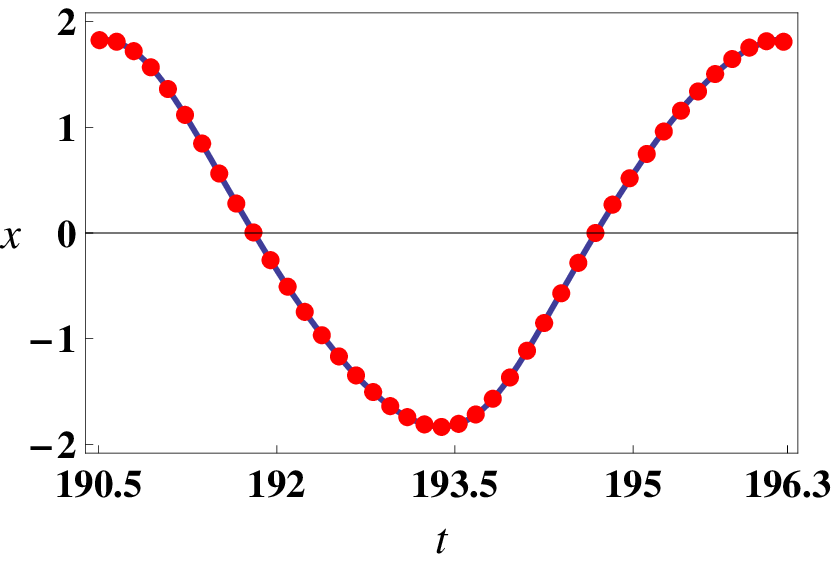}\label{f:6e}} &
\qquad\qquad &
\subfigure[]{\includegraphics[height=0.3\linewidth,width=0.38\linewidth]{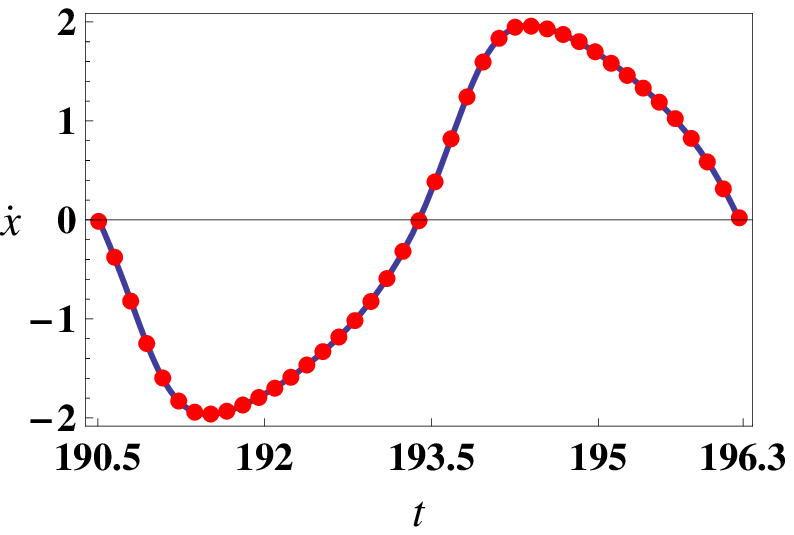}\label{f:6f}}
\end{tabular}
\end{center}
\caption{Improving the approximation shown in Figure \ref{f:3a} discussed in
Section \ref{Reg RLE}. Here \textbf{Solid} lines represent exact solutions and
\textbf{Dotted} lines represent approximate solutions.\medskip\newline$\left(
a\right)  $ Comparison of the exact solution $x\left(  t\right)  $ of $\left(
\text{\ref{RFRL Eqn}}\right)  $ and the RGM approximation given by $\left(
\text{\ref{RFRL RG New Sol}}\right)  $ for $185<t<200.$ $\left(  b\right)  $
Comparison of $x\left(  t+2.7\right)  $ and the RGM approximation given by
$\left(  \text{\ref{RFRL RG New Sol}}\right)  $ showing better agreement for
$185<t<200$. $\left(  c\right)  $ Comparison of exact error $x\left(
t+2.7\right)  -x_{0}\left(  t\right)  $ and its estimate of $\psi_{x}\left(
t\right)  $ given by $\left(  \text{\ref{Error2 Estimate in x}}\right)  $ in
$\left[  r_{1},r_{3}\right]  $. $\left(  d\right)  $ Comparison of exact error
$\dot{x}\left(  t+2.7\right)  -\dot{x}_{0}\left(  t\right)  $ and its estimate
of $\psi_{y}\left(  t\right)  $ given by $\left(
\text{\ref{Error2 Estimate in y}}\right)  $ in $\left[  r_{1},r_{3}\right]  $.
$\left(  e\right)  $ Comparison of $x\left(  t+2.7\right)  $ and its estimate
$x_{0}\left(  t\right)  +\psi_{x}\left(  t\right)  $. $\left(  f\right)  $
Comparison of $\dot{x}\left(  t+2.7\right)  $ and its estimate $\dot{x}%
_{0}\left(  t\right)  +\psi_{y}\left(  t\right)  $.}%
\label{f:6}%
\end{figure}

As in Example A, we next estimate percentage errors along $x$ and $y$
directions as $E_{x}=1.26\%$ and $E_{y}=1.08\%$ where $\tilde{a}=\max
_{r_{1}\leq t\leq r_{3}}\left\vert x\left(  t+2.7\right)  \right\vert
=1.83249,\text{ and }\ \max_{r_{1}\leq t\leq r_{3}}\left\{  x\left(
t+2.7\right)  -\left(  x_{0}\left(  t\right)  +\psi_{x}\left(  t\right)
\right)  \right\}  =0.0230739$, and similarly for $E_{y}$, and hence accuracy
of $O(98.74\%)$ and $O(98.92\%)$ is achieved along $x$ and $y$ directions
respectively. The estimated limit cycle with $x_{0}\left(  t\right)  +\psi
_{x}\left(  t\right)  $ and $\dot{x}_{0}\left(  t\right)  +\psi_{y}\left(
t\right)  $ is compared with exact limit cycle in Figure \ref{f:7}%
.\begin{figure}[ptb]
\begin{center}
\includegraphics[width=0.3\linewidth,width=0.3\linewidth]{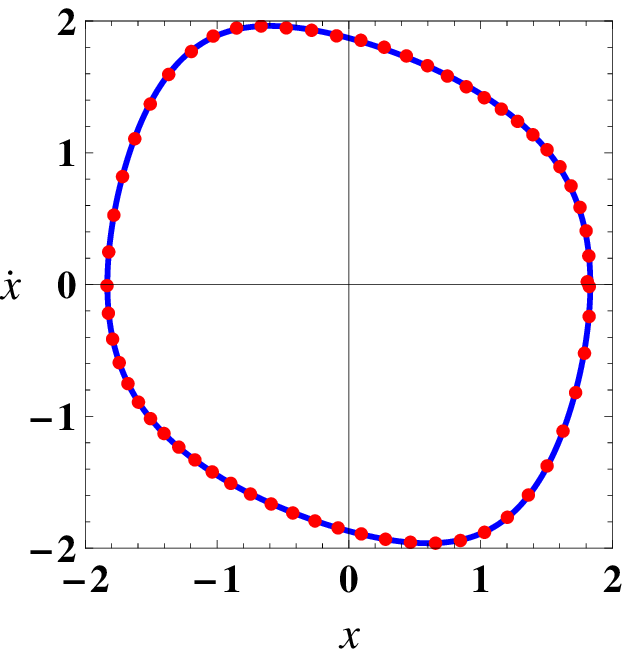}
\end{center}
\caption{The solid line represents numerically computed exact limit cycle as
discussed in Section \ref{Reg RLE} and the dotted line represents the
estimated limit cycle by $x_{0}\left(  t\right)  +\psi_{x}\left(  t\right)  $
and $\dot{x}_{0}\left(  t\right)  +\psi_{y}\left(  t\right)  $.}%
\label{f:7}%
\end{figure}

\subsection{Discussion and Comparison\label{Sec Discussion}}

The above two examples clearly demonstrate the role and significance of
dynamical time scales (\ref{dynnorm}) introduced along $x$ and $y$ in the
phase plane. In Figure \ref{f:4} - Figure \ref{f:7}, one clearly notices the
distributions of estimated data points undergoing a process of condensations
and rarefactions over specific intervals along the time axis that correspond
to slow and fast motions in the relaxation oscillation concerned. In Figure
\ref{f:5}, in particular, the branch of relaxation oscillation undergoing slow
building up accommodates more data points, when a sparser set of data points
is presented along the branch as the system relaxes very fast. This fact is
demonstrated explicitly in Figures \ref{f:4c}-\ref{f:4f} presenting different
error estimates relevant for the relaxation oscillation of Figure \ref{f:5}.
One notices clearly denser set of data points in the neighbourhood of $t=196$
and a sparser set near $t=197$ in all these four figures. Similar
condensations and rarefactions are also observed near $t=198$ and $t=195$
respectively, reflecting periodicity of the original orbit. Similar behaviours
are also observed in Figure \ref{f:6} and Figure \ref{f:7}. Such oscillations
in the density of estimated data points involving dynamical time scales along
independent phase space variables may be considered as an interesting novel
aspect of the present asymptotic formalism. The very high level of efficiency,
defined by $\left(  \text{\ref{effx}}\right)  $ and $\left(  \text{\ref{effy}%
}\right)  $, is obtained by a very simple computational scheme, developed in
the present formalism based, however, on some novel dynamical insights, has
the added advantage of visual demonstrations of the existence of slow fast
motions in the nonlinear oscillations. The present study is a specific
application of the nonlinear $SL\left(  2,%
\mathbb{R}
\right)  $ asymptotic analysis which, in fact, have a wider range of
applications \cite{ds15, palit_comparative_2016, dss18, dsr20} and
philosophical implications \cite{dpnw}. The computations presented involve
only a few $($actually one or two only$)$ harmonic terms when conventional
approaches based on either homotopy analysis method \cite{Cui-NA-2018} or
generalized averaging method etc. \cite{Luo_book, Xu-Luo-2019} require $10$ or
more harmonic terms to attain equivalent level of efficiency (accuracy).

\section{Conclusion\label{Sec Conclusion}}

The novel framework of $SL(2,%
\mathbb{R}
)$ invariant asymptotic structures is presented in the context of a nonlinear
oscillatory system. The significance and nontrivial applications of these
asymptotic structures are explained for evaluating highly efficient phase
portraits (accuracy level exceeding $98\%$) of period 1 relaxation cycles of
strongly nonlinear and singularly perturbed Rayleigh equations with periodic
external forcing. We argue that $SL(2,%
\mathbb{R}
)$ asymptotic structures is powerful enough to upgrade standard RG
computations on nonlinear orbits into highly efficient computations. New
asymptotic structures are introduced by implementing continuous deformation of
conventional linear slow scales such as $t_{n}=\varepsilon^{n}%
t,\ 0<\varepsilon<1$ into nonlinear \emph{dynamic} scales of the form
$T_{n}=t_{n}\sigma(t_{n})$, for a nonlinear deformation factor $\sigma(t_{n})$
respecting some well defined $SL(2,%
\mathbb{R}
)$ constraints. The concept and role of $(i)$ self dual deformation exponents
as well as $(ii)$ its extension to dynamic time scale deformations in a slow
-fast system involving very slow and very fast evolutionary directions, are
explained in detail in the context of efficient computations of period 1
relaxation oscillations. The formalism presented here is robust and has a
wider field of mathematical and other applications \cite{dpnw}. Application of
the formalism to more general dually related deformation exponents will be
considered separately in addressing period doubling bifurcations to chaos
systematically for nonlinear oscillatory problems.

\section*{Acknowledgments}

The second author (Dhurjati Prasad Datta) wish to thank IUCAA, Pune for
offering a visiting Associateship. He is also thankful to IUCAA Centre for
Astronomy Research and Development (ICARD), University of North Bengal for
offering its facilities.

\section*{Appendix}

It is known that the forced Rayleigh equation $\left(  \text{\ref{RFRL Eqn}%
}\right)  $ has a unique limit cycle around the critical point $\left(
0,0\right)  $ in the $x,\dot{x}$ phase plane. So, the shape of the limit cycle
does not depend on the initial condition for the system $\left(
\text{\ref{RFRL Eqn}}\right)  $. Thus, for the sake of completeness we choose
the initial condition
\begin{equation}
x\left(  t_{0}\right)  =\alpha,\ \dot{x}\left(  t_{0}\right)  =\beta\label{IC}%
\end{equation}
where, $\left(  \alpha,\beta\right)  \neq\left(  0,0\right)  $. We are trying
to find limit cycle solution when the frequency $\Omega$ of the external
periodic force is close to the frequency of the system. Therefore, we assume%
\begin{equation}
\Omega=1+\varepsilon\sigma,\ \varepsilon>0 \label{Resonant Frequency}%
\end{equation}
and the limit cycle solution can be written as the perturbative series%
\begin{equation}
x\left(  t\right)  =x_{0}\left(  t\right)  +\varepsilon\ x_{1}\left(
t\right)  +\varepsilon^{2}x_{2}\left(  t\right)  +\cdots\text{.} \label{Pert}%
\end{equation}
Using $\left(  \text{\ref{Resonant Frequency}}\right)  $ and $\left(
\text{\ref{Pert}}\right)  $ in $\left(  \text{\ref{RFRL Eqn}}\right)  $ and
comparing coefficients of different orders of $\varepsilon$ we obtain the
following equations.
\begin{subequations}
\label{RFRL Formulation}%
\begin{align}
\text{Zero-th Order}  &  :\ddot{x}_{0}(t)+\omega^{2}x_{0}(t)=0\text{ with
}x\left(  t_{0}\right)  =\alpha,\ \dot{x}\left(  t_{0}\right)  =\beta
,\label{Order Zero DE}\\
\varepsilon\text{ Order}  &  :\ddot{x}_{1}(t)+\omega^{2}x_{1}(t)+\frac{1}%
{3}\left(  \dot{x}_{0}(t)\right)  {}^{3}-\dot{x}_{0}(t)-F\cos\omega t=0,\quad
x_{1}\left(  t_{0}\right)  =0,\ \dot{x}_{1}\left(  t_{0}\right)
=0,\label{Order One DE}\\
\varepsilon^{2}\text{ Order}  &  :\ddot{x}_{2}(t)+\omega^{2}x_{2}(t)+\dot
{x}_{1}(t)\left(  \dot{x}_{0}(t)\right)  {}^{2}-\dot{x}_{1}(t)+F\ t\ \sigma
\sin\omega t=0,\text{\quad}x_{2}\left(  t_{0}\right)  =0,\ \dot{x}_{2}\left(
t_{0}\right)  =0,\label{Order Two DE}\\
&  \cdots\ \cdots\ \cdots\ \cdots\qquad\cdots\ \cdots\ \cdots\ \cdots
\qquad\cdots\ \cdots\ \cdots\ \cdots\qquad\cdots\ \cdots\ \cdots
\ \cdots\nonumber
\end{align}
The general solution of the zero-th order equation $\left(
\text{\ref{Order Zero DE}}\right)  $ can be written as%
\end{subequations}
\begin{equation}
x_{0}\left(  t\right)  =Ae^{i\left(  t-t_{0}\right)  }+A^{\ast}e^{-i\left(
t-t_{0}\right)  }, \label{Zoro Sol}%
\end{equation}
where $A$ is a complex constant. Although $A$ can be determined from the
initial conditions, we keep $A$ undetermined because the computational steps
in the RGM will ultimately replace $A$ by its counterpart $\mathcal{A}$, which
will not remain a constant of motion. The solution of the $\varepsilon$ order
equation $\left(  \text{\ref{Order One DE}}\right)  $ is%
\begin{equation}
x_{1}\left(  t\right)  =\left(  \frac{1}{24}iA^{3}+\frac{1}{2}\left(
A-A^{2}A^{\ast}\right)  \left(  t-t_{0}\right)  -\frac{F\ i}{4}\left(
t-t_{0}\right)  \right)  e^{i\left(  t-t_{0}\right)  }-\frac{1}{24}%
iA^{3}e^{3i\left(  t-t_{0}\right)  }+c.c. \label{One Sol}%
\end{equation}
where $c.c.$ designates the complex conjugate of the preceding expression.
Similarly, we can continue to higher order solution to generate the naive
perturbative solution of the system $\left(  \text{\ref{RFRL Eqn}}\right)  $.

In the second step of the RGM, we renormalize the integration constant $A$ and
create a counterpart $\mathcal{A}$ as%
\[
A=\mathcal{A+}a_{1}\mathcal{\varepsilon}+a_{2}\varepsilon^{2}+\cdots
\]
where, the coefficients $a_{1}$, $a_{2}$, \ldots\ are chosen to absorb the
homogeneous parts of the solution in different orders of $\varepsilon$.
Choosing%
\begin{equation}
a_{1}=-\dfrac{1}{24}i\mathcal{A}^{3} \label{a1Val}%
\end{equation}
we get,%
\begin{equation}
x\left(  t\right)  =\mathcal{A}e^{i\left(  t-t_{0}\right)  }+\varepsilon
\left\{  \dfrac{1}{2}\left(  \mathcal{A}-\mathcal{A}^{2}\mathcal{A}^{\ast
}\right)  \left(  t-t_{0}\right)  e^{i\left(  t-t_{0}\right)  }-\dfrac
{F\ i}{4}\left(  t-t_{0}\right)  e^{i\left(  t-t_{0}\right)  }-\dfrac{1}%
{24}i\mathcal{A}^{3}e^{3i\left(  t-t_{0}\right)  }\right\}  +c.c.+O\left(
\varepsilon^{2}\right)  . \label{SolHalfRG1}%
\end{equation}

In the third step, we need to differentiate the expression containing
$e^{i\left(  t-t_{0}\right)  },\ e^{-i\left(  t-t_{0}\right)  },\ \left(
t-t_{0}\right)  e^{i\left(  t-t_{0}\right)  }$ and $\left(  t-t_{0}\right)
e^{-i\left(  t-t_{0}\right)  }$ in $\left(  \text{\ref{SolHalfRG1}}\right)  $
with respect to $t_{0}$ and thereafter substituting $t_{0}=t$ the resultant
expression is equated to zero to get the RG condition $\left(
\text{\ref{RG Eq}}\right)  $. The terms related to higher harmonic are not
involved in RG condition. Thus, $\left(  \text{\ref{RG Eq}}\right)  $ gives%
\begin{equation}
\left.  \frac{\partial\mathcal{A}}{\partial t_{0}}\right\vert _{t_{0}%
=t}=\mathcal{A\ }i+\varepsilon\ \dfrac{1}{2}\left(  \mathcal{A}-\mathcal{A}%
^{2}\mathcal{A}^{\ast}\right)  -\dfrac{F\ i}{4}\varepsilon+O\left(
\varepsilon^{2}\right)  \label{RG Condition}%
\end{equation}
where, $\mathcal{A}$ is not a constant of motion, rather it is a function of
time $t$. Using $\left(  \text{{\ref{Complex to Polar}}}\right)  $ in $\left(
\text{\ref{RG Condition}}\right)  $ where $\Omega=1+\varepsilon\sigma$ and
comparing the real and imaginary components we get after simplification
\begin{subequations}
\label{RG Eq RFRL}%
\begin{align}
\dfrac{dR}{dt} &  =\varepsilon\dfrac{R}{2}\left(  1-\dfrac{R^{2}}{4}\right)
-\dfrac{F}{2}\varepsilon\sin\theta+O\left(  \varepsilon^{2}\right)
\label{RG Eq1 RFRL New}\\
\dfrac{d\theta}{dt} &  =1-\dfrac{F}{2R}\varepsilon\cos\theta+O\left(
\varepsilon^{2}\right)  \text{.}\label{RG Eq2 RFRL New}%
\end{align}

The corresponding limit cycle solution is finally obtained from $\left(
\text{\ref{SolHalfRG1}}\right)  $ by putting $t_{0}=t$ and using the polar
decomposition of $\mathcal{A}$ given by (\ref{Complex to Polar}), in the form%

\end{subequations}
\begin{equation}
x\left(  t\right)  =R\left(  t\right)  \cos\theta(t) +\varepsilon\frac
{R^{3}\left(  t\right)  }{96}\sin3 \theta(t) \text{.} \label{RG Sol}%
\end{equation}

Taking%
\[
\theta=\phi+\Omega\ t=\phi+\left(  1+\varepsilon\sigma\right)  t
\]
we write the RG flow equations $\left(  \text{\ref{RG Eq RFRL}}\right)  $ as
\begin{subequations}
\label{RG Eq RFRL New}%
\begin{align}
\dfrac{dR}{dt}  &  =\varepsilon\dfrac{R}{2}\left(  1-\dfrac{R^{2}}{4}\right)
-\dfrac{F}{2}\varepsilon\sin\left(  \phi+\Omega\ t\right)  +O\left(
\varepsilon^{2}\right) \\
\dfrac{d\phi}{dt}+\varepsilon\sigma &  =-\dfrac{F}{2R}\varepsilon\cos\left(
\phi+\Omega\ t\right)  +O\left(  \varepsilon^{2}\right)  \text{.}%
\end{align}

For steady state motion $\frac{dR}{dt}=0$ and $\frac{d\phi}{dt}=0$ so that
$\left(  \text{\ref{RG Eq RFRL}}\right)  $ give%

\end{subequations}
\begin{align*}
\frac{R}{2}\left(  1-\frac{R^{2}}{4}\right)   &  =\frac{F}{2}\sin\left(
\phi+\Omega\ t\right) \\
R\sigma &  =-\frac{F}{2}\cos\left(  \phi+\Omega\ t\right)  \text{.}%
\end{align*}
Squaring and then adding we get,%
\begin{align*}
&  \left.  \frac{R^{2}}{4}\left(  1-\frac{R^{2}}{4}\right)  ^{2}+R^{2}%
\sigma^{2}=\frac{F^{2}}{4}\right.  \text{.}%
\end{align*}
Taking%
\[
\rho=\frac{R^{2}}{4}\text{ and }k=F
\]
we get,%
\begin{equation}
\rho\left(  1-\rho\right)  ^{2}+4\sigma^{2}\rho=\frac{k^{2}}{4},
\label{Freq-Response Eq}%
\end{equation}
which is exactly same as equation $\left(  4.3.15\right)  $ in the book of
Nayfeh and Mook \cite{nayfey_nonlinear_1995}, Page 205. The region giving the
stable limit cycle is the common region of $\Delta>0$ and $\rho>\frac{1}{2}$
in $\rho\sigma$ plane. If we take%
\begin{equation}
\varepsilon=0.5,\ k=F=0.5,\ \sigma=0.1 \label{RFRL Parameter Val}%
\end{equation}
then $\left(  4.3.15\right)  $ gives three values of $\rho$ as%
\[
\rho=6.8915\times10^{-2},\ 0.80632\text{ and}\ 1.1248\text{.}%
\]
It is discussed in \cite{nayfey_nonlinear_1995} that for $k>0$ the largest
value of $\rho$ always gives stable limit cycle. Therefore, we take the
largest value $\rho=1.1248$ so that%
\[
\Delta=8.408\,1\times10^{-2}>0.
\]
Therefore, this choice of parameters given by $\left(
\text{\ref{RFRL Parameter Val}}\right)  $ satisfy both the convergence
criteria%
\[
\Delta>0\text{ and }\rho>\frac{1}{2}\text{.}%
\]
Thus, for this values of the parameters the limit cycle given by $\left(
\text{\ref{RG Sol}}\right)  $ is shown in Figure \ref{f:8a} by dotted line
along with the numerical (exact) limit cycle in solid line for the Rayleigh
equation $\left(  \text{\ref{RFRL Eqn}}\right)  $.

\begin{figure}[ptb]
\begin{center}
\centering%
\begin{tabular}
[c]{rcr}%
\subfigure[]{\includegraphics[height=0.3\linewidth]{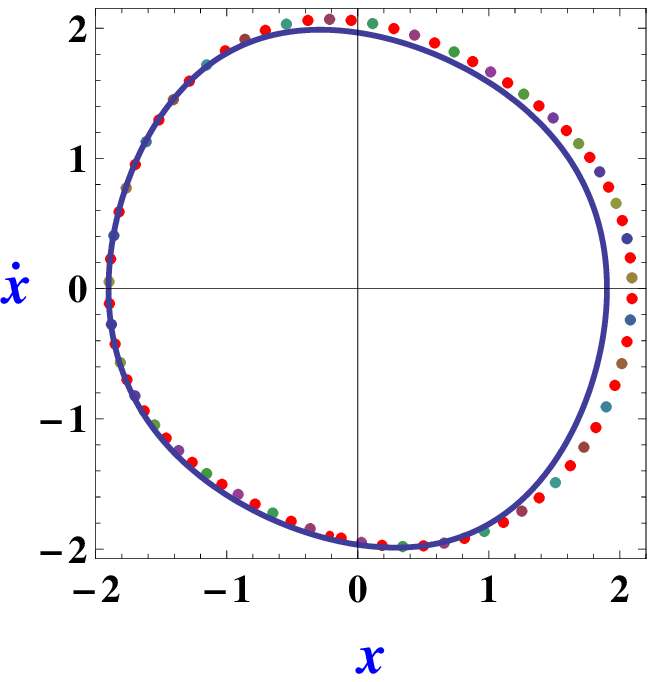}\label{f:8a}} &
\qquad\qquad &
\subfigure[]{\includegraphics[height=0.3\linewidth]{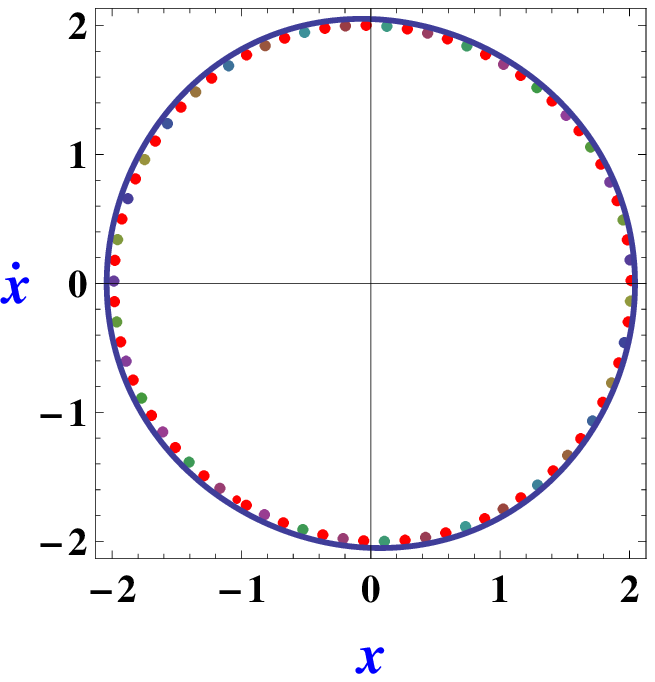}\label{f:8b}}
\end{tabular}
\end{center}
\caption{{Comparison of limit cycles given by the approximate solution
}$\left(  \text{\ref{RG Sol}}\right)  ${ using RGM (in dotted line) with the
exact numerical limit cycle (in solid line) for }$\left(  a\right)  $
$\varepsilon=0.5,\ k=F=0.5,$ $\sigma=0.1$ and $\left(  b\right)
\ \varepsilon=0.1,\ k=F=0.25,$ $\sigma=0.05$.}%
\label{f:8}%
\end{figure}Next, for%
\begin{equation}
\varepsilon=0.1,\ k=F=0.25,\ \sigma=0.05 \label{RFRL Parameter Val2}%
\end{equation}
we have%
\[
\rho=1.5971\times10^{-2},\ 0.91599\text{ and }1.068\text{.}%
\]
As discussed in \cite{nayfey_nonlinear_1995} for $k>0$ the largest value of
$\rho$ always gives stable limit cycle. Therefore, we take the largest value
$\rho=1.068$ so that%
\[
\Delta=3.99\times10^{-2}>0\text{.}%
\]
Therefore, this choice of parameters given by $\left(
\text{\ref{RFRL Parameter Val2}}\right)  $ satisfy both the convergence
criteria%
\[
\Delta>0\text{ and }\rho>\frac{1}{2}\text{.}%
\]
Thus, for this values of the parameters the limit cycle given by $\left(
\text{\ref{RG Sol}}\right)  $ is shown in Figure \ref{f:8b} by dotted line
along with the numerical (exact) limit cycle in solid line for the Rayleigh
equation $\left(  \text{\ref{RFRL Eqn}}\right)  $.

\bibliographystyle{IEEEtran}
\bibliography{PDR-Bib}

\end{document}